\documentclass[conference]{IEEEtran}
\IEEEoverridecommandlockouts

\usepackage[noadjust]{cite}
\usepackage{amsmath,amssymb,amsfonts}
\usepackage{graphicx}
\usepackage{textcomp}
\usepackage{xcolor}
\def\BibTeX{{\rm B\kern-.05em{\sc i\kern-.025em b}\kern-.08em
    T\kern-.1667em\lower.7ex\hbox{E}\kern-.125emX}}

\usepackage{tikz}
\usepackage{algorithm}
\usepackage{algpseudocode}
\usepackage{amsthm}
\usepackage{stmaryrd}
\usepackage{booktabs}
\usepackage{listings}
\usepackage[table]{xcolor}

\usepackage{hyperref}
\usepackage{xspace}

\def\ufp{\mathsf{ufp}}
\def\ulp{\mathsf{ulp}}
\def\nsb{\mathsf{nsb}}
\def\max{\mathsf{max}}
\def\min{\mathsf{min}}
\def\R{\mathbb{R}}

\def\round{{\uparrow_\sim^p}}
\def\emin{{e_{min}}}

\newtheorem{theorem}{Theorem}
\theoremstyle{definition}
\newtheorem{definition}{Definition}
\newtheorem{property}{Property}
\usepackage{makecell}

\usepackage{todonotes}

\newif\ifdraft
\draftfalse

 \newcommand{\fp}{floating-point\xspace} 
 \newcommand{\FP}{Floating-point\xspace}

\begin{document}

\title{FPScan: An Automated Constraint-Based Analyzer for Floating-Point Anomaly Detection
}

\author{\IEEEauthorblockN{Julien Bortolussi, Dorra Ben Khalifa, and Pierre-Loïc Garoche}
\IEEEauthorblockA{
    \textit{Fédération ENAC ISAE-SUPAERO ONERA, Université de Toulouse,}
    Toulouse, France\\
    \{julien.bortolussi,dorra.ben-khalifa,pierre-loic.garoche\}@enac.fr
}
}
\definecolor{mygreen}{RGB}{34,139,34}    
\definecolor{myorange}{RGB}{255,140,0}   
\definecolor{myred}{RGB}{220,20,60}      

\maketitle

\begin{abstract}
Writing error-free floating-point programs is a challenging task, especially for
programmers who lack a strong background in numerical analysis and
rounding-error propagation. State-of-the-art techniques typically aim to bound
such errors using static or dynamic analysis. However, only a few tools
explicitly address critical floating-point pitfalls such as absorption and
catastrophic cancellation. These anomalies represent situations in which
rounding errors are significantly amplified, causing the semantics of the
finite-precision computation to deviate substantially from the real-number
semantics. 
In this article, we present FPScan, a novel tool to formally define and detect
both catastrophic cancellation and absorption in floating-point programs. Our
approach starts with a custom static analyzer based on abstract interpretation
to infer the order of magnitude of all program variables. This magnitude
information is then used to build a set of first-order constraints that model
error propagation and numerical precision within the program. Finally, we employ
an off-the-shelf SMT solver to determine whether the program exhibits any of
these critical numerical pitfalls. Experiments were conducted on
FPBench, a well-known benchmark suite of floating-point programs, to evaluate
the effectiveness of our tool. We also present a comparison with state-of-the-art tools regarding soundness and analysis time.
\end{abstract}

\begin{IEEEkeywords}
Floating-point arithmetic, Static analysis, Catastrophic cancellation, Absorption, Constraint generation, SMT solver.
\end{IEEEkeywords}

\section{Introduction}
\FP arithmetic is the standard approach for approximating real-number computations in modern computing systems \cite{IEEE754}. By representing numbers in a scientific-notation format with a finite significand and a bounded exponent, it supports a large dynamic range while keeping relative errors small.
However, this representation implies that \fp numbers are not uniformly distributed over the real line. As a consequence, certain operations may introduce significant numerical errors. Two well-known examples are absorption and catastrophic cancellation. Absorption occurs when adding numbers with vastly different magnitudes, causing the smaller value to be lost due to rounding. Catastrophic cancellation arises when subtracting nearly equal numbers, amplifying previously accumulated rounding errors.

In practice, many \fp expressions are implemented by developers without specialized training in numerical analysis, which makes diagnosing such numerical issues difficult. As a result, developers often rely on iterative trial-and-error modifications guided by test inputs, an approach that is both inefficient and unreliable \cite{panchekha_automatically_2015}.

\FP experts typically describe these phenomena as numerical pitfalls that explain the origins of many \fp errors \cite{goldberg_what_1991,Odyssey_2023,panchekha_automatically_2015}. However, relatively little work has investigated how these pitfalls could be systematically exploited, either to assist developers in identifying problematic computations or to automatically generate more accurate expressions.


In this article, we introduce FPScan\footnote{Software artifact available at \url{https://github.com/JBortolussi/FPScan}}, a tool for soundly and automatically detecting \fp pitfalls. 
FPScan is based on the principle that pitfalls are closely related to the relative orders of magnitude of operands and the magnitude of the associated rounding errors. 
To exploit this property, FPScan first performs a static range analysis to infer interval bounds for program variables and intermediate expressions. These bounds are then used to approximate the orders of magnitude of values involved in computations. FPScan abstracts these orders of magnitude using integer variables and generates constraints that characterize conditions under which numerical pitfalls, such as absorption and catastrophic cancellation, may occur. The resulting constraints are checked using the Z3 SMT solver to determine whether such situations are detected within the inferred ranges.
We evaluate FPScan using the FPBench benchmark suite~\cite{damouche2016toward} and compare it with FPChecker~\cite{laguna_fpchecker_2019,laguna_fpchecker_2022}, a state-of-the-art dynamic tool for detecting \fp cancellations, and with a bitblasting-based approach. Our results show that FPScan achieves high precision while remaining computationally efficient.

 The three main contributions of this article are: 
 \begin{enumerate} 
    \item An efficient constraint-based method for detecting absorption and  catastrophic cancellation in \fp computations (including denormalized numbers) based on order-of-magnitude reasoning (sections~\ref{section::background} and  \ref{section::detection}) .
    \item FPScan, a tool based on  static analysis  that combines interval range analysis with SMT solving to automatically detect \fp pitfalls (Section  \ref{section::detection}) .
    \item An experimental evaluation on the FPBench benchmark suite, including comparisons with FPChecker and a bitblasting-based approach (Section~\ref{section::evaluation}). It demonstrates the efficiency and accuracy of FPScan.
\end{enumerate}

\section{Background on Floating-Point Arithmetic}
\label{section::background}
In this section, we review the fundamentals of \fp arithmetic, with a focus on its key concepts and main sources of numerical errors.

\subsection{Floating-Point Numbers}

\FP arithmetic approximates real-number computations using a finite number of digits, making it suitable for computer systems. According to the IEEE 754 Standard~\cite{IEEE754}, a \fp number is represented as shown
\begin{equation} \small  x = (-1)^s \times b_0.b_1b_2...b_{p-1} \times \beta^e \enspace  ,  \label{def::float} \end{equation}
 where $s$ denotes the sign bit, $b_0.b_1b_2...b_{p-1}$ is the significand written in base $\beta$ (here we consider $\beta = 2$), $e$ is the exponent, and $p$ is the precision. The base $\beta$, the precision $p$, and the bounds of the exponent $e_{min}$ and $e_{max}$ together define a \fp format.

The IEEE 754 Standard defines several \fp formats that differ in precision and exponent range. The most common formats are $binary32$, $binary64$, and $binary128$, which provide $24$, $53$, and $113$ bits of precision $p$, respectively.

The \fp representation of real numbers is not unique. For instance, with $\beta = 10$, the numbers $9.0$ and $0.9 \times 10^1$ represent the same value. To address this ambiguity, \fp numbers are \textit{normalized} so that their significand is maximized, i.e., in this example, $9.0$.

According to Eq.~(\ref{def::float}), the leading digit of the significand $(b_0)$ is equal to one for all \textit{normalized} numbers. Numbers that do not satisfy this property are called \textit{denormalized}.

The process of selecting which of the two nearest \fp numbers represents a real value is called rounding. The IEEE 754 Standard \cite{IEEE754} defines several rounding modes.
In this work, we focus exclusively on the most common setting: the $\round$ rounding mode, which rounds a number to its nearest \fp neighbour. 

\subsection{Absorption and Catastrophic Cancellation in \FP Arithmetic}
Rounding errors arise in all \fp operations and may accumulate, leading to significant loss of accuracy. Certain situations, known as \fp pitfalls, can further amplify this effect. 
The specific focus of this work is to identify such pitfalls: absorption and catastrophic cancellation.
While these phenomena are widely discussed in the literature, they are rarely given formal definitions. 
Let us first introduce them informally before giving a more precise definition in 
Section~\ref{section::contribution::pitfalls}.

\subsubsection{Absorption}\label{def::absorption::txt} occurs when a
large number is added to a small number, causing the contribution of the smaller
number to be similar to a rounding error. Although only a regular rounding error is introduced
in this situation, it can lead to substantial problems, as the result differs
from what the programmer intended. 
Consider the program in
Figure~\ref{prgm::example} with precision \(p = 24\) ($binary32$), \(x = 2^{26}\), and \(y
= 1\). Line~1 computes \(z = \round(2^{26} + 1) = 2^{26}\). Although \(z\) is
affected only by a rounding error, the information provided by $y$ is lost.

\begin{figure}[h] \small
\centering
\vspace{-0.5em}
\hrule
\vspace{0.3em}
\begin{tabular}{ll}\small
1: & \textbf{Require} $x \in [1,2^{26}] \wedge y \in [1,10]$ \\
2: & $z \gets x+y$ \\
3: & $\_\_result\_\_ \gets z-x$ \\
4: & \textbf{Return} $\_\_result\_\_$ \\
\end{tabular}\hrule
\caption{\normalsize \textit{Running example} program.}
\label{prgm::example}
\vspace{0.3em}
\end{figure}

\subsubsection{Catastrophic Cancellation}
\label{def::cancellation::txt} 
 occurs when subtracting two nearby numbers whose most significant bits coincide and cancel each other out making the result significantly smaller than the operands. Consequently, any error affecting one of the operands has a proportionally larger impact on the result.
Again, consider the program in Figure~\ref{prgm::example}, with precision \(p = 24\) ($binary32$), \(x = 2^{26}\), \(y = 1\), and \(z = 2^{26}\). Line~2 computes \(\_\_result\_\_ = \round(2^{26} - 2^{26}) = 0\) instead of \(1\). In this case, the subtraction between \(z\) and \(x\) leads to catastrophic cancellation, which reveals the pitfall caused by the absorption in Line~1. This large relative error in the result can then propagate significantly through subsequent computations.

Unlike the definitions given in~\cite{LamHS13,laguna_fpchecker_2022}, in this work we rely on the formalization of Unit in the First Place (\(\ufp\)), Unit in the Last Place (\(\ulp\)), and Number of Significant Bits (\(\nsb\)). In the next section, we formally define absorption and catastrophic cancellation in terms of these quantities
at the abstract-semantics level, expressed as the satisfiability of a constraint. 

\section{Constraint-Based Detection of Floating-Point Absorption and
Catastrophic Cancellation}
\label{section::detection}
In this section, we present the approach implemented in FPScan.
The overall workflow of FPScan is illustrated in Figure~\ref{fig::workflow}. It first parses the \fp program and performs a range analysis to infer bounds on variables and intermediate expressions. These bounds are then used to generate constraints characterizing the conditions under which absorption and catastrophic cancellation may occur. The resulting constraints are checked independently using an SMT solver to determine whether such situations are feasible. 
\normalsize
\begin{figure}[tb]
        \includegraphics[width=\columnwidth]{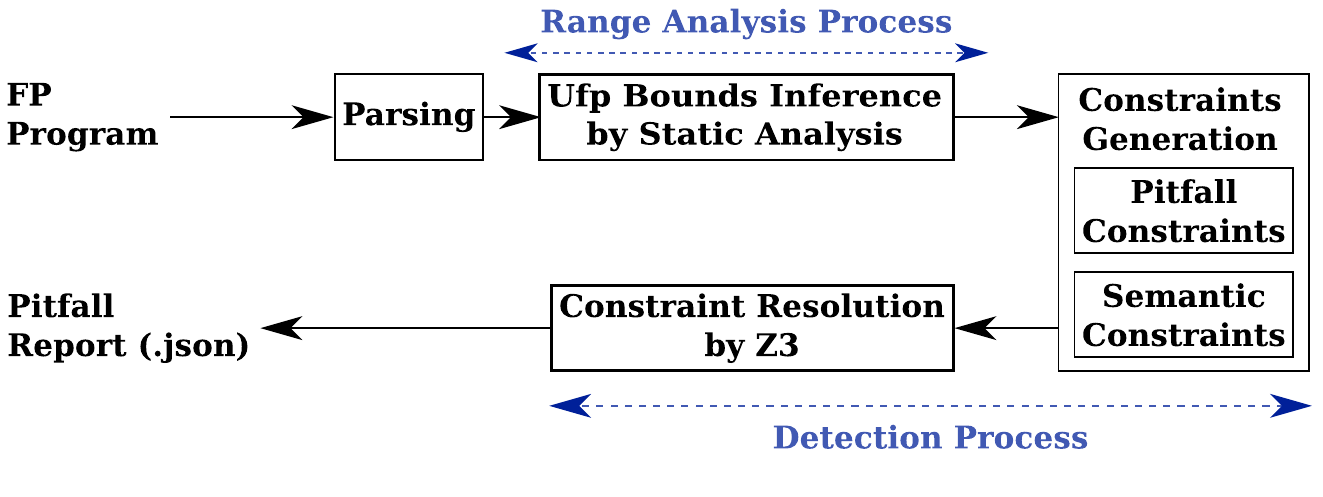}
        \vspace{-2em}
    \caption{FPScan's workflow.}
    \label{fig::workflow}
    \vspace{-1em}
\end{figure}

\FP pitfalls depend on the orders of magnitude of values and their errors. Therefore, both can be abstracted by their orders of magnitude, yielding integer constraints that are faster to solve. 

\smallskip
\noindent\textbf{Notation:}
In this section, we use the following notations.
Let $p \in \mathbb{N}$
and $e_{min}\in \mathbb{Z}$ denote the precision and the minimum exponent of the \fp format. Let $\mathbb{F}_p$ denote the set of \fp numbers in this format.
Let $\mathbb{V}$ be a set of variables and let $\Sigma = (\mathbb{V} \rightarrow \mathbb{R} \times \mathbb{F}_p)$ be the set of environments, i.e., the set of machine states mapping each variable to its real and \fp values. 
To simplify the notation, 
let $\sigma_{\mathbb{R}}$ and $\sigma_{\mathbb{F}_p}$ be the projections of $\sigma \in \Sigma$ onto $\mathbb{R}$ and $\mathbb{F}_p$, respectively. 
For all variables $v \in \mathbb{V}$, we define $\tilde{v}=\sigma_{\mathbb{F}_p}(v)$.

\subsection{Quantifying Magnitude, Precision, and Error}
\label{section::ufp}
We now introduce three quantities that characterize the order of magnitude of a \fp number, its precision, and its associated error.

\begin{definition}[Unit in the First Place ($\ufp$)]
    \begin{equation}\small
        \begin{array}{rl}
            \ufp\colon \mathbb{R}   & \longrightarrow \mathbb{Z} \\
            x                       & \longmapsto  \left\{\begin{array}{lll}
                                            \min\{i \in \mathbb{Z} \mid 2 ^{i+1} > |x|\} & \textnormal{if}\quad x\neq0\\
                                            \emin - p + 1& \textnormal{else.}
                                        \end{array}\right.
        \end{array}
        \label{def::ufp}
    \end{equation}
\end{definition}
The $\ufp$ corresponds to the order of magnitude of the number:
\begin{equation}
    \forall x \in \R, \forall n \in \mathbb{Z}, \ufp(x) \leq n \Longleftrightarrow |x| < 2^{n+1}\, .
    \label{def::ufp::bound}
\end{equation}

\begin{definition}[Unit in the Last Place ($\ulp$)]
  \begin{equation}\small
    \begin{array}{rl}
        \ulp: \mathbb{F}_p  & \longrightarrow \mathbb{Z} \\
        \tilde{x}           & \longmapsto \max(\ufp(\tilde{x}) - p + 1, e_{min} - p + 1)\, \enspace .
    \end{array}
  \end{equation}
  \label{def::ulp}\vspace{-1.5em}
\end{definition}

The $\ulp$ corresponds to the exponent of the least significant bit and is valid for all floating-point formats. This definition requires special care when handling denormalized numbers.

Let us note that the $\ulp$ is closely related to rounding errors~\cite{goldberg_what_1991}, as shown in the following equation:
\begin{equation}\small
    \forall x \in \mathbb{R},~|x - \round(x)| \leq 2^{\ulp(\round(x)) - 1} \, .
    \label{def::ulp::error}
\end{equation}

\begin{definition}[Number of Significant Bits ($\nsb$)] 
Let $x_v
\in \mathbb{V}$ a variable
and $\sigma \in \Sigma$ an environment.
\begin{equation}\small
\begin{aligned}
\quad
\nsb_\sigma : \mathbb{V} &\longrightarrow \mathbb{Z} \\
x_v &\longmapsto \ufp\!\big(\sigma_{\mathbb{F}_p}(x_v)\big)
- \ufp\!\big(\sigma_{\mathbb{R}}(x_v) - \sigma_{\mathbb{F}_p}(x_v)\big)\, .
\end{aligned}
\label{def::nsb::eq}
\end{equation}
\label{def::nsb}
\end{definition}
\vspace{-1.5em}

The $\nsb$ corresponds to the number of bits whose weight exceeds the rounding error. Consequently, it can take negative values when rounding error exceeds the magnitude of the value itself.

\subsection{Concrete Program Semantics}
\label{section::language}

The syntax of the language accepted by the FPScan parser is given in Figure~\ref{fig::program_semantic}. It supports basic arithmetic operators, mathematical functions, conditionals, and loops. In practice, FPScan takes a subset of FPCore~\cite{damouche2016toward} as input and translates it into this internal syntax.
\begin{figure}[tb]
\hrule
\small
\vspace{0.2cm}
\resizebox{\linewidth}{!}{
$
\begin{array}{rl}
\textbf{Variables:}     & (x_v, y_v, z_v, t_v) \in \mathbb{V}^4 \\[6pt]
\textbf{Operators:}     & \lozenge \in \{+, -, \times, \div\} \\[6pt]
\textbf{Functions:}     & \mathsf{f} \in \{cos, sin, atan, sqrt\} \\[6pt]
\textbf{Comparisons:}   & \rhd \in \{>, \geq, =, \neq\}\\[6pt]
\textbf{Statements:}    & \begin{array}[t]{lcl}
                            s & :=  & s;s \\
                              & |   & x_v = n, ~ n \in \mathbb{R} \\
                              & |   & x_v = y_v \\
                              & |   & x_v = y_v \lozenge z_v    \\
                              & |   & x_v = f(y_v)              \\
                              & |   & \mathbf{if} ~ t_v \rhd 0 ~ \mathbf{then} ~ s_{then} ~ \mathbf{else} ~ s_{else} ~ \mathbf{endif}   \\
                              & |   & \mathbf{while} ~ t_v \rhd 0 ~ \mathbf{do} ~ s ~ \mathbf{done }   \\
                              & |   & \mathbf{noop}
                        \end{array}
\end{array}
$
}
\vspace{0.2cm}
\hrule
\caption{FPScan language.}
\label{fig::program_semantic}
\vspace{-1em}
\end{figure}
As highlighted by Titolo~et~al.~\cite{TitoloMFMM24}, most static analysis tools
support only statically bounded loops. FPScan is no exception and therefore only
supports bounded loops, with the maximum number of iterations specified by the
user. During a pre-processing stage, these loops are unrolled and replaced with
a sequence of nested conditional statements that execute the loop body as long
as the loop condition holds. The user-defined bound determines the depth of this
conditional chain. Additionally, we assume that both real and floating‑point
executions follow the same control‑flow path in conditional statements, i.e.,
test instabilities are not modeled. Handling unbounded loops and
unstable tests is left for future work.

The semantic function $\mathbb{S} \llbracket s \rrbracket : \Sigma \rightarrow \Sigma$ defines the semantics of a statement $s$ by mapping a machine state to another machine state according to the execution of $s$. Its definition is given in Figure~\ref{fig::concrete_semantic}.
Here, $\sigma \in \Sigma$ denotes a machine state; $s_1$ and $s_2$ denote statements; $x_v$, $y_v$, $z_v$, and $t_v$ denote variables; $n \in \mathbb{R}$ denotes a real-valued constant; $\rhd \in \{>, \geq, =, \neq\}$ denotes a comparison operator; and $\mathsf{f} \in \{cos, sin, atan, sqrt\}$ denotes a supported function.

In Eq.~\eqref{def::concrete_C}, we define $\mathcal{C} \llbracket s \rrbracket : \mathcal{P}(\Sigma) \rightarrow \mathcal{P}(\Sigma)$, the collecting semantic function that returns the set of reachable machine states resulting from executing a statement over a set of input states. Here, $\mathcal{P}(\Sigma)$ denotes the powerset of $\Sigma$.

\begin{figure} [tb]
\hrule \vspace{0.2cm} \resizebox{\linewidth}{!}{ \setlength{\arrayrulewidth}{0.4pt} \small  \begin{tabular}{l} $\mathbb{S} \llbracket s_1;~s_2 \rrbracket (\sigma) ~:=~\mathbb{S} \llbracket s_2 \rrbracket \big(\mathbb{S} \llbracket s_1 \rrbracket (\sigma)\big)$\\[6pt] $\mathbb{S} \llbracket x_v = n \rrbracket (\sigma) ~:=~\sigma \big[ x_v \leftarrow (n, \round (n)) \big]$ \\[6pt] $\mathbb{S} \llbracket z_v = x_v \lozenge y_v \rrbracket (\sigma) ~:=~\sigma \left[ z_v \leftarrow \left(\begin{array}{l} \sigma_{\mathbb{R}}(x_v) \lozenge \sigma_{\mathbb{R}}(y_v), \\ \round (\sigma_{\mathbb{F}_p}(x_v) \lozenge \sigma_{\mathbb{F}_p}(y_v)) \end{array}\right)\right]$\\[12pt] $\mathbb{S} \llbracket \mathbf{if}~t_v \rhd 0~\mathbf{then}~s_{then}~\mathbf{else}~s_{else}~\mathbf{endif} \rrbracket(\sigma)$ \\ \quad $:= \mathsf{if}~\sigma_\mathbb{R}(t_v) \rhd 0~ \mathsf{then}~\mathbb{S}\llbracket s_{then} \rrbracket~ \mathsf{else}~\mathbb{S}\llbracket s_{else} \rrbracket $ \\[12pt] $\mathbb{S} \llbracket x_v = f(y_v) \rrbracket (\sigma) ~:=~\sigma\left[ x_v \leftarrow ( f(\sigma_{\mathbb{R}}(y_v)), \round (\mathsf{f} (\sigma_{\mathbb{F}_p}(y_v))) )\right]$ \\[6pt] $\mathbb{S} \llbracket x_v = y_v \rrbracket (\sigma) ~:=~\sigma \left[ x_v \leftarrow \sigma(y_v) \right]$\\[6pt] $\mathbb{S} \llbracket \mathbf{noop} \rrbracket (\sigma)~:=~\sigma$ \end{tabular} } \vspace{0.2cm} \hrule \caption{Concrete semantics of the language.} \label{fig::concrete_semantic} 
\vspace{-1.5em}
\end{figure}

\begin{equation}\small
    \forall X \in \mathcal{P}(\Sigma),~\mathcal{C}\llbracket s \rrbracket (X) = \{\mathbb{S} \llbracket s \rrbracket (\sigma) ~|~ \sigma \in X\}.
    \label{def::concrete_C}
\end{equation}


\subsection{Approximated Constraint-based Program Semantics}
FPScan abstracts the semantic function $\mathcal{C}$, defined in Eq.~\eqref{def::concrete_C}, using constraints over $\ufp$, $\ulp$, and $\nsb$. Let $\mathbb{L}$ denote the set of the constraints. A constraint \(l : \Sigma \rightarrow \{\text{true}, \text{false}\}\) 
is a boolean-valued function
over a machine state.
To formalize this abstraction, we introduce a concretization-based abstract interpretation~\cite{cousot1977abstract}. The concretization function \(\gamma\), defined in Eq.~\eqref{def::cstr::gamma}, maps each constraint to the set of machine states that satisfy it.
\begin{equation}\small
    \begin{array}{rl}
        \gamma\colon \mathbb{L} & \longrightarrow   \mathcal{P}(\Sigma) \\
        l                       & \longmapsto       \{\sigma \in \Sigma | l(\sigma) = true\}
    \end{array}
    \label{def::cstr::gamma}
\end{equation}
The function $\mathcal{C}^{\#}\llbracket s \rrbracket : \mathbb{L} \rightarrow
\mathbb{L}$ represents the abstract program semantics. To be considered correct,
it must satisfy Eq.~\eqref{def::abstract}. In other words, it must not omit
any reachable machine states. 
\begin{equation}\small
   \forall l \in \mathbb{L},~\mathcal{C} \llbracket s \rrbracket (\gamma (l)) \subseteq \gamma (\mathcal{C}^{\#} \llbracket s \rrbracket (l))
    \label{def::abstract}
\end{equation}
\normalsize

\subsubsection{Range analysis}
FPScan first performs an interval analysis to detect and reject programs that may produce overflows, NaNs, or divisions by zero. The inferred intervals, of the form $[\underline{x}, \overline{x}]$, provide bounds on program variables and are then 
used to derive bounds on their $\ufp$ values using Eq.~\eqref{def::ufp::bound}. To this end, we introduce the following three bounding constraints:
\[\small
\left\{
\begin{aligned}\small 
\mathsf{UpperBound}\llbracket x_v \rrbracket
&= \lambda \sigma.\;
\ufp(\sigma_{\mathbb{F}_p}(x_v))
\leq
\ufp\!\big(\max(|\underline{x}|,|\overline{x}|)\big),
\\[0.3em]
\mathsf{LowerBound}\llbracket x_v \rrbracket
&= \lambda \sigma.\;
\ufp\!\big(\min(|\underline{x}|,|\overline{x}|)\big)
\leq
\ufp(\sigma_{\mathbb{F}_p}(x_v)),
\\[0.3em]
\mathsf{Bound}\llbracket x_v \rrbracket
&= \lambda \sigma.\;
\left\{
\begin{aligned}
&\mathsf{UpperBound}\llbracket x_v \rrbracket(\sigma)\\
&\wedge\;
\mathsf{LowerBound}\llbracket x_v \rrbracket(\sigma)
\  \text{if } 0 \notin [\underline{x},\overline{x}] .
\end{aligned}
\right.
\end{aligned}
\right.
\]

\subsubsection{Constraint generation}

The constraints used to abstract program executions are derived from the concrete semantics of the language. 
For the remainder of this section, let $x_v$, $y_v$, $z_v$ and $t_v$ be program variables in $\mathbb{V}$.
Let $s_1$, $s_2$, $s_{then}$, and $s_{else}$ be statements in the language.
We now define the constraints associated with each statement of the language.
Since our constraints are expressed solely in terms of $\ufp$, $\ulp$, and $\nsb$, they do not capture the sign of the operands. As a consequence, addition and subtraction are treated uniformly at the constraint level.

\begin{definition}[Addition/Subtraction]
For all $l \in \mathbb{L}$, the abstraction of the statement  $z_v = x_v \pm y_v$ is defined as 
    \[\small
        \resizebox{0.95\linewidth}{!}{$
        \mathcal{C}^\# \llbracket z_v = x_v \pm y_v \rrbracket (l) :=
            \lambda \sigma. \left[\begin{array}{l}
                l(\sigma) \wedge \mathsf{Bound} \llbracket z_v \rrbracket (\sigma) \\
                \wedge\ \nsb_\sigma(z_v) \geq \ufp(\tilde{z}) - \ufp_e
            \end{array}\right.
        $}
     \]
    with
    \vspace{-1em}
    \begin{equation*}\small
        \ufp_e = \max\left(\begin{array}{l}
                \max\left(\begin{array}{l}
                    \ufp(\tilde{x}) - \nsb_\sigma(x_v),\;\\
                    \ufp(\tilde{y}) - \nsb_\sigma(y_v)
                \end{array}\right) + 2,\\[6pt]
                \ulp(\tilde{z})
            \end{array} \right) \enspace .
    \end{equation*}
    \label{cstr::add::theorem}
\end{definition}

\begin{definition}[Multiplication]\label{cstr::mult::theorem}
$\forall$ $l \in \mathbb{L}$,
the abstraction of the statement $z_v = x_v \times y_v$ is defined as
\[
\small
\mathcal{C}^\# \llbracket z_v = x_v \times y_v \rrbracket (l)
:= \lambda \sigma.
\left[
\begin{aligned}
&l(\sigma)
\wedge \mathsf{Bound} \llbracket z_v \rrbracket (\sigma) \\
&\wedge\;
\nsb_\sigma(z_v) \geq \ufp(\tilde{z}) - \ufp_e
\end{aligned}
\right.
\]
with
 \vspace{-1em}
\begin{equation*}\small
\resizebox{1\columnwidth}{!}{$
\ufp_e =
\max\left(
\begin{array}{l}
\ufp(\tilde{x}) + \ufp(\tilde{y})
- \min\big(\nsb_\sigma(x_v), \nsb_\sigma(y_v)\big) + 3,\\[4pt]
\ufp(\tilde{x}) + \ufp(\tilde{y})
- \nsb_\sigma(x_v) - \nsb_\sigma(y_v) + 3,\\[4pt]
\ulp(\tilde{z})
\end{array}
\right)\enspace.
$}
\end{equation*}
\end{definition}

\begin{definition}[Division]\label{cstr::div::theorem}
     $\forall l \in \mathbb{L}$, the abstraction of the
    statement $z_v = x_v \div y_v$ is defined as 
    \[
    \small
        \resizebox{0.95\linewidth}{!}{$
        \mathcal{C}^\# \llbracket z_v = x_v \div y_v \rrbracket (l) :=
            \lambda\sigma. \left[\begin{array}{l}
                l(\sigma) \wedge \mathsf{Bound} \llbracket z_v \rrbracket (\sigma) \\
                ~ \wedge \left(
                    \begin{array}{c}
                        \nsb_\sigma(z_v) \geq \ufp(\tilde{z}) - \ufp_e  \\
                        \vee    \\
                        \nsb_\sigma(y_v) \leq 1
                    \end{array}
                \right)
            \end{array}\right.
        $}
    \]
\vspace{-1em}
\noindent
    with\vspace{0.2em}
    \begin{equation*}\small
        \resizebox{1\linewidth}{!}{$
        \ufp_e = \max\left(
        \begin{array}{l}
            \ufp(\tilde{x}) - \ufp(\tilde{y}) - \min(\nsb_\sigma(x_v), \nsb_\sigma(y_v)) + 4 \\
            \ulp(\tilde{z}) - 1
        \end{array}
        \right) \enspace .
        $}
    \end{equation*}
\end{definition}

\begin{definition}[Cos, Sin, and Atan]\label{cstr::trigo::theorem}
  $\forall l \in \mathbb{L}$,  the abstraction of the
    statement $z_v = f(x_v)$ with $f \in \{\text{cos}, \text{sin}, $ $\text{atan}\}$ is defined as 
    \vspace{-.5em}
    \[
    \small
        \resizebox{0.95\linewidth}{!}{$
        \mathcal{C}^\# \llbracket z_v = f(x_v) \rrbracket (l) := \lambda\sigma.
        \left[\begin{array}{l}
            l(\sigma) \wedge \mathsf{Bound} \llbracket z_v \rrbracket (\sigma) \\
            \wedge ~ \nsb_\sigma(z_v) \geq \ufp(\tilde{z}) - \ufp_e\\
        \end{array}\right.
        $}
        \label{cstr::trigo::theorem::Eq}
    \vspace{-1em}
    \]
    \vspace{0.2em}
    \noindent
    with
    \begin{equation*}\small
        \ufp_e =    \max\left(\begin{array}{l}
                        \ufp(\tilde{x}) - \nsb_\sigma(x_v) + 1   \\
                        \ulp(\tilde{z}) - 1
                    \end{array}\right) \enspace .
    \end{equation*}
\end{definition}

\begin{definition}[Sqrt]\label{cstr::sqrt::theorem}
    $\forall l \in \mathbb{L}$, the abstraction of the
    statement $z_v = \sqrt{x_v}$  is defined as
    \[
    \small 
        \mathcal{C}^\# \llbracket z_v = \sqrt{x_v   } \rrbracket (l) := \lambda \sigma . \left[
            \begin{array}{l}
                l(\sigma) \wedge \mathsf{Bound} \llbracket z_v \rrbracket (\sigma) \\
                \wedge ~ \nsb_\sigma(z_v) \geq \ufp(\tilde{z}) - \ufp_e
            \end{array}
        \right.
        \label{cstr::sqrt::theorem::Eq}
    \]
    with
    \begin{equation*} \small
        \ufp_e = \max (\frac{\ufp(\tilde{x}) - \nsb_\sigma(x_v) + 1}{2}, \ulp(\tilde{z}) - 1) \enspace .
    \end{equation*}
\end{definition}

\begin{definition}[Assign Constant]\label{cstr::asnc::theorem}
    Let $n \in \mathbb{R}$ be a constant, $\forall l \in \mathbb{L}$, the abstraction of the statement $x_v := n$ is defined as
    \[
    \small
        \mathcal{C}^\#\llbracket x_v := n \rrbracket (l) =
            \lambda\sigma. \left[\begin{array}{l}
                l(\sigma) \\
                \wedge ~ \mathsf{Bound}\llbracket x_v \rrbracket (\sigma) \\
                \wedge ~ \nsb(x_v) \geq p
            \end{array}\right. \enspace .
    \]
\end{definition}

\begin{definition}[Assign Variable]\label{cstr::asnv::theorem}
$\forall l \in \mathbb{L}$, the abstraction of the statement $x_v := y_v$ is defined as
    \[\small
        \mathcal{C}^\#\llbracket x_v := y_v \rrbracket (l) =
            \lambda\sigma. l(\sigma) \wedge x_v = y_v \enspace .
    \]
\end{definition}

\begin{definition}[Sequence]\label{cstr::seq::theorem}
    Let $s_1$ and $s_2$ denote two statements. The abstract semantics of the sequence is defined as 
    \[\small
        \forall l \in \mathbb{L},~\mathcal{C}^\#\llbracket s_1;~s_2 \rrbracket (l) := \mathcal{C}^\# \llbracket s_2 \rrbracket (\mathcal{C}^\# \llbracket s_1 \rrbracket (l)) \enspace .
        \label{cstr::sequence}
    \]
\end{definition}

\begin{definition}[If Then Else]\label{cstr::ite::theorem}
    Let $\rhd$ be a
    comparison operator, then the abstract branching condition is defined as 
    \[   
            \begin{array}{r}
                \forall l \in \mathbb{L}, \mathcal{C}^\# \llbracket \mathbf{if}~t_v \rhd 0~\mathbf{then}~s_{then}\mathbf{else}~s_{else}~\mathbf{endif} \rrbracket (l) := \\
                \hfill \qquad \mathcal{C}^\#\llbracket s_{then}\rrbracket(l)(\sigma) \vee \mathcal{C}^\#\llbracket s_{else}\rrbracket(l)(\sigma)
            \end{array}
    \]
\end{definition}

\begin{theorem}[Soundness of abstract semantics $\mathcal{C}^\#$]
The abstract semantics $\mathcal{C}^\#$, 
is a sound over-approximation of the concrete semantics $\mathcal{C}$, \textit{i.e.} satisfies Eq.~\eqref{def::abstract}.
\end{theorem}

\begin{proof}
Soundness proof is performed by structural induction on the program constructs.
\emph{Abstract math. functions soundness:} 
Defs~\ref{cstr::add::theorem}--
\ref{cstr::sqrt::theorem} constraints 
are all made of three parts. The first propagates previous constraints. The second bounds the ufp of variables using $\mathsf{Bound} \llbracket . \rrbracket$. It is sound because the interval analysis producing the bound is. The third bounds the $\nsb$. Let us develop the derivation of the bound on the $\nsb$ for the (Addition/Subtraction) constraint.
The soundness proof of Defs~\ref{cstr::mult::theorem}--
\ref{cstr::sqrt::theorem} are provided in Appendix. 
\emph{Abstract addition soundness:}
Let $\varepsilon_x$, $\varepsilon_y$, and $\varepsilon_z$ denote the \fp errors: for $v \in \{x,y,z\}$, 
$\varepsilon_v = v - \tilde{v} \Leftrightarrow \tilde{v} = v - \varepsilon_v \Leftrightarrow v = \tilde{v} + \varepsilon_v$.
Let \( e_+ \) be the
rounding error introduced by the addition.
We can derive 
$\tilde{z} = \round(\tilde{x} + \tilde{y}) = \tilde{x} + \tilde{y} + e_+ = (x - \varepsilon_x) + (y - \varepsilon_y) + e_+
           = 
           {x + y}
           - (
           {\varepsilon_x + \varepsilon_y - e_+}
           )$, characterizing  \( \varepsilon_z \) as \(\varepsilon_x + \varepsilon_y - e_+\).
 The error $|\varepsilon_z|$ is bounded using the triangle inequality    
 $|\varepsilon_z| \leq |\varepsilon_x| + |\varepsilon_y| + |e_+|$ while
$|\varepsilon_x|$ and $|\varepsilon_y|$ are bounded by 
Eq.~\eqref{def::ufp::bound} and $e_+$ by 
$|e_+| \leq 2^{\ulp(\tilde{z}) - 1}$ using Eq.~\eqref{def::ulp::error}.
We then have 
$    |\varepsilon_z| < 2^{\ufp(\varepsilon_x) + 1} + 2^{\ufp(\varepsilon_y) + 1} + 2^{\ulp(\tilde{z}) - 1} < 
2^{\max(\ufp(\varepsilon_x) + 1, \ufp(\varepsilon_y) + 1, \ulp(\tilde{z})-1) + 2}$.
    Using Eq.~\eqref{def::ufp::bound} we obtain
$\ufp(\varepsilon_z) \leq \max(\ufp(\varepsilon_x) + 1, \ufp(\varepsilon_y) + 1, \ulp(\tilde{z})-1) + 1$.
\emph{Abstract control flow soundness:} Let us now focus on the statements of the language described in Defs~\ref{cstr::asnc::theorem}--
\ref{cstr::ite::theorem}. 
Def.~\ref{cstr::asnc::theorem} is sound because to round a real number to make it fit the \fp format adds no more than a rounding error, hence $\nsb \geq p$.
Def.~\ref{cstr::asnv::theorem} states that no additional error is introduced copying a variable and Def.~\ref{cstr::seq::theorem} translates constraint propagation.
Def.~\ref{cstr::ite::theorem} states the control flow may take either path of a conditional statement without any consideration for the condition. This is an overapproximation because any possible execution of the program necessarily follows one of the branches. Therefore, this abstraction is sound.
\end{proof}

\subsection{A Taxonomy of Floating-Point Pitfalls}
\label{section::contribution::pitfalls}

In this section, we formally define absorption and catastrophic cancellation as constraints.
These numerical issues arise specifically in the context of addition and subtraction operations. We recall that addition and subtraction are treated in the same way in our analysis.
An addition or subtraction operation may exhibit two types of absorption: either the left-hand-side operand or the right-hand-side operand may be absorbed.

\begin{definition}[Absorption]
$\forall x_v, y_v \in \mathbb{V}$, \( x_v \) is \emph{absorbed} by \( y_v \) in the computation \( x_v \pm y_v \) if and only if
\[
\small
\ufp(\tilde{x}) \leq \ulp(\tilde{y}) - 1 \enspace .
\]
\label{def::absorption}
\end{definition}
\vspace{-1em}

Def.~\ref{def::absorption} defines the condition under which \( x_v \) is absorbed by \( y_v \).
Intuitively, this occurs when the magnitude of \( x \) (measured by its \( \ufp \)) is smaller than the resolution of \( y \), i.e., the magnitude of a rounding unit at \( y \) given by \( \ulp(\tilde{y}) \).
Swapping the roles of \( x_v \) and \( y_v \) yields the symmetric condition corresponding to the absorption of \( y_v \) by \( x_v \).
We now formalize our definition of catastrophic cancellation.

\begin{definition}[Catastrophic Cancellation]
Let $x_v, y_v, z_v \in \mathbb{V}$ be three variables such that \( z_v = x_v \pm y_v \).
\( z_v \) is affected by a catastrophic cancellation if and only if
\[\small
\ufp(\tilde{z}) \leq \ufp(\tilde{x}) - \nsb_\sigma(x_v) \;\vee\; \ufp(\tilde{z}) \leq \ufp(\tilde{y}) - \nsb_\sigma(y_v) \enspace .
\]
\label{def::cancellation}
\end{definition}
\vspace{-1em}

Def.~\ref{def::cancellation} defines the conditions under which a catastrophic cancellation occurs during the addition or subtraction of $x_v$ and $y_v$. Intuitively, this occurs when the magnitude of the result (measured by its $\ufp$) is smaller than the magnitude of the error in either of the operands (measured by its $\nsb$).

\begin{property}[Detection of same sign arguments]
If the errors are sufficiently small, no spurious catastrophic cancellation is detected for same sign additions, or different sign subtraction.
\end{property}

\begin{proof}
Let us consider $z_v = x_v + y_v$ with $\tilde{x} \geq 0$ and $\tilde{y} \geq 0$.
By symmetry, let us focus on $x_v$ and assume $\nsb_\sigma(x_v) > 0$. Since $\tilde{y} \geq 0$, we have $\tilde{x} \leq \tilde{x} + \tilde{y}$. By monotonicity of rounding, $\tilde{x} = \round(\tilde{x}) \leq \round(\tilde{x} + \tilde{y}) = \tilde{z}$. Moreover, $\tilde{x} > |\varepsilon_x|$ and therefore $\tilde{z} > |\varepsilon_x|$. Using Eq.~\eqref{def::ufp::bound} and the definition of $\nsb$, $2^{\ufp(\tilde{z}) + 1} > 2^{\ufp(\tilde{x}) - \nsb_\sigma(x_v)}$ and thus $\ufp(\tilde{z}) + 1 > \ufp(\tilde{x}) - \nsb_\sigma(x_v)$ which means that the formula used to characterize catastrophic cancellation is false.
\end{proof}


\subsection{Detection of Floating-Point Pitfalls Using an SMT solver}

FPScan constructs a set of constraints to be checked using an SMT solver. Each constraint abstracts the set of reachable machine states exhibiting a specific numerical pitfall. One constraint is generated for each pitfall that may occur in the program.

Each constraint consists of two components. The first characterizes the set of reachable machine states and is derived using $\mathcal{C}^\#$. The second identifies machine states exhibiting a specific pitfall, as defined in Defs.~\ref{def::absorption}~and~\ref{def::cancellation}.

Again, we consider the program illustrated in Figure~\ref{prgm::example}. Line~1 may be affected by three numerical pitfalls: absorption of $x$ by $y$, absorption of $y$ by $x$, and catastrophic cancellation. Each case is analyzed independently.
We first consider the absorption of $x$ by $y$. FPScan constructs a constraint characterizing the machine states reachable after Line~1 and satisfying the absorption condition. This constraint is obtained by combining Defs.~\ref{cstr::add::theorem} and~\ref{def::absorption}. The resulting formula is submitted to the Z3~\cite{Z3} SMT solver. Since the constraint is unsatisfiable, this case is proven impossible.
The same procedure is applied to the absorption of $y$ by $x$. In this case, the constraint is satisfiable, and FPScan reports a warning indicating that the pitfall may occur.
The same analysis is repeated for catastrophic cancellation. Line~2 is treated in the same way. Note that the set of reachable machine states after Line~2 is constrained by both Line~1 and Line~2.
The report produced by FPScan is shown in Figure~\ref{fig::example::result}. It was generated in less than $0.15$ seconds. It contains two warnings on Line~2 and one warning on Line~1, which implies that the remaining three potential pitfalls are proven impossible. In particular, catastrophic cancellation cannot occur on Line~1.

The report highlights the root causes of numerical inaccuracies in Figure~\ref{fig::example::result}. For large values of $x$, $y$ is absorbed on Line~1, introducing a small error and, more importantly, a semantic gap that is later revealed by cancellation on Line~2.
However, the absorption reported on Line~2 is spurious, as it would imply that $z$ becomes larger than $x$, which is impossible. This issue arises from the non-relational nature of the constraints used to bound $\ufp$.
This is expected due to the sound yet incomplete nature of our analysis. We will see in the experimental evaluation that the conservativeness of the method is limited, giving very accurate results.
{
\begin{figure}[tb]
\hrule
\scriptsize
\begin{lstlisting}[language=Python, basicstyle=\small\ttfamily]
File "..." => z = x + y:
	ABSORPTION: y absorbed by x
File "..." => __result__ = z - x:
	CANCELLATION
	ABSORPTION: z absorbed by x
\end{lstlisting}
\vspace{0.2em}
\hrule
\caption{FPScan Output for Figure~\ref{prgm::example}.}
\label{fig::example::result}
\end{figure}

}

\section{Experimental Evaluation}
\label{section::evaluation}
We now investigate the following research questions by evaluating FPScan on the
FPBench benchmark suite~\cite{damouche2016toward}:

\begin{itemize}
    \item \textbf{RQ1}: To what extent can floating point pitfalls be detected using an abstraction based on $\ufp$, $\ulp$, and $\nsb$?
    \item \textbf{RQ2}: How does the runtime performance of our constraint solving approach compare to that of existing state-of-the-art tools?
\end{itemize}

\begin{table}[bt]
\caption{Comparison of FPScan with the  {bitblasting} based ground truth.}

\vspace{-5pt}

\label{tab::results}
\scriptsize
\setlength{\tabcolsep}{2pt}
\renewcommand{\arraystretch}{1.0}

\resizebox{\columnwidth}{!}{%
\begin{tabular}{l|c|cccc|c|cccc}
\textbf{Program} 
& \multicolumn{5}{|c}{\textbf{Absorption}} 
& \multicolumn{5}{|c}{\textbf{Cancellation}}\\

\midrule

& \cellcolor{gray!45}\rotatebox{90}{\textbf{FPScan}} 
& \cellcolor{gray!45}\textbf{FP} 
& \textbf{TP} 
& \cellcolor{gray!25}\textbf{FN} 
& \textbf{TN} 
& \cellcolor{gray!45}\rotatebox{90}{\textbf{FPScan}} 
& \cellcolor{gray!45}\textbf{FP} 
& \textbf{TP} 
& \cellcolor{gray!25}\textbf{FN} 
& \textbf{TN}\\

\textbf{triangle}                            & 0     & \cellcolor{gray!45} \textbf{0}   & {\textbf{\color{gray!75}0}}    & \cellcolor{gray!25} 0 & 10    & 0     & \cellcolor{gray!45} \textbf{0}    & {\textbf{\color{gray!75}0}}   & \cellcolor{gray!25} 0 & 5     \\
\textbf{bspline3}                            & 0     & \cellcolor{gray!45} \textbf{0}   & {\textbf{\color{gray!75}0}}   & \cellcolor{gray!25} 0 & 2     & 0     & \cellcolor{gray!45} \textbf{0}    & {\textbf{\color{gray!75}0}} & \cellcolor{gray!25} 0 & 1     \\
\textbf{test06\_sums4\_\_sum(x2)}            & 12    & \cellcolor{gray!45} \textbf{0}   & 12                    & \cellcolor{gray!25} 0 & 16    & 6     & \cellcolor{gray!45} \textbf{0}    & 6                     & \cellcolor{gray!25} 0 & 2     \\
\textbf{test05\_nonlin1\_\_test2}            & 0     & \cellcolor{gray!45} \textbf{0}   & {\textbf{\color{gray!75}0}}    & \cellcolor{gray!25} 0 & 2     & 0     & \cellcolor{gray!45} \textbf{0}    & {\textbf{\color{gray!75}0}}    & \cellcolor{gray!25} 0 & 1     \\
\textbf{test05\_nonlin1\_\_r4}               & 0     & \cellcolor{gray!45} \textbf{0}   & {\textbf{\color{gray!75}0}}    & \cellcolor{gray!25} 0 & 4     & 0     & \cellcolor{gray!45} \textbf{0}    & {\textbf{\color{gray!75}0}}    & \cellcolor{gray!25} 0 & 2     \\
\textbf{test03\_nonlin2}                     & 2     & \cellcolor{gray!45} \textbf{0}   & 2                     & \cellcolor{gray!25} 0 & 6     & 1     & \cellcolor{gray!45} \textbf{0}    & 1                     & \cellcolor{gray!25} 0 & 2     \\
\textbf{test02\_sum8}                        & 0     & \cellcolor{gray!45} \textbf{0}   & {\textbf{\color{gray!75}0}}    & \cellcolor{gray!25} 0 & 14    & 0     & \cellcolor{gray!45} \textbf{0}    & {\textbf{\color{gray!75}0}}    & \cellcolor{gray!25} 0 & 7     \\
test01\_sum3                                 & 4     & 1                                & 3                     & \cellcolor{gray!25} 0 & 15    & 5     & 2                                 & 3                     & \cellcolor{gray!25} 0 & 3     \\
\textbf{nonlin1}                             & 1     & \cellcolor{gray!45} \textbf{0}   & 1                     & \cellcolor{gray!25} 0 & 2     & 0     & \cellcolor{gray!45} \textbf{0}    & {\textbf{\color{gray!75}0}}    & \cellcolor{gray!25} 0 & 1     \\
\textbf{nonlin2}                             & 0     & \cellcolor{gray!45} \textbf{0}   & {\textbf{\color{gray!75}0}}    & \cellcolor{gray!25} 0 & 4     & 0     & \cellcolor{gray!45} \textbf{0}    & {\textbf{\color{gray!75}0}}    & \cellcolor{gray!25} 0 & 2     \\
\textbf{verhulst}                            & 0     & \cellcolor{gray!45} \textbf{0}   & {\textbf{\color{gray!75}0}}    & \cellcolor{gray!25} 0 & 2     & 0     & \cellcolor{gray!45} \textbf{0}    & {\textbf{\color{gray!75}0}}    & \cellcolor{gray!25} 0 & 1     \\
\textbf{x\_by\_xy}                           & 0     & \cellcolor{gray!45} \textbf{0}   & {\textbf{\color{gray!75}0}}    & \cellcolor{gray!25} 0 & 2     & 0     & \cellcolor{gray!45} \textbf{0}    & {\textbf{\color{gray!75}0}}    & \cellcolor{gray!25} 0 & 1     \\
\textbf{predatorPrey}                        & 0     & \cellcolor{gray!45} \textbf{0}   & {\textbf{\color{gray!75}0}}    & \cellcolor{gray!25} 0 & 2     & 0     & \cellcolor{gray!45} \textbf{0}    & {\textbf{\color{gray!75}0}}    & \cellcolor{gray!25} 0 & 1     \\
\textbf{hypot(x2)}                           & 0     & \cellcolor{gray!45} \textbf{0}   & {\textbf{\color{gray!75}0}}    & \cellcolor{gray!25} 0 & 4     & 0     & \cellcolor{gray!45} \textbf{0}    & {\textbf{\color{gray!75}0}}    & \cellcolor{gray!25} 0 & 2     \\
\textbf{sec4\_example}                       & 0     & \cellcolor{gray!45} \textbf{0}   & {\textbf{\color{gray!75}0}}    & \cellcolor{gray!25} 0 & 4     & 0     & \cellcolor{gray!45} \textbf{0}    & {\textbf{\color{gray!75}0}}    & \cellcolor{gray!25} 0 & 2     \\
\textbf{intro\_example}                      & 1     & \cellcolor{gray!45} \textbf{0}   & 1                     & \cellcolor{gray!25} 0 & 2     & 0     & \cellcolor{gray!45} \textbf{0}    & {\textbf{\color{gray!75}0}}    & \cellcolor{gray!25} 0 & 1     \\
\textbf{intro\_example\_mixed}               & 0     & \cellcolor{gray!45} \textbf{0}   & {\textbf{\color{gray!75}0}}    & \cellcolor{gray!25} 0 & 2     & 0     & \cellcolor{gray!45} \textbf{0}    & {\textbf{\color{gray!75}0}}    & \cellcolor{gray!25} 0 & 1     \\
\textbf{smartRoot}                           & 0     & \cellcolor{gray!45} \textbf{0}   & {\textbf{\color{gray!75}0}}    & \cellcolor{gray!25} 0 & 12    & 0     & \cellcolor{gray!45} \textbf{0}    & {\textbf{\color{gray!75}0}}   & \cellcolor{gray!25} 0 & 6     \\
\textbf{sqrt\_add}                           & 0     & \cellcolor{gray!45} \textbf{0}   & {\textbf{\color{gray!75}0}}    & \cellcolor{gray!25} 0 & 4     & 0     & \cellcolor{gray!45} \textbf{0}    & {\textbf{\color{gray!75}0}}    & \cellcolor{gray!25} 0 & 2     \\
\textbf{sqroot}                              & 4     & \cellcolor{gray!45} \textbf{0}   & 4                     & \cellcolor{gray!25} 0 & 8     & 0     & \cellcolor{gray!45} \textbf{0}    & {\textbf{\color{gray!75}0}}    & \cellcolor{gray!25} 0 & 4     \\
\textbf{squareRoot3(x2)}                     & 6     & \cellcolor{gray!45} \textbf{0}   & 6                     & \cellcolor{gray!25} 0 & 12    & 2     & \cellcolor{gray!45} \textbf{0}    & 2                     & \cellcolor{gray!25} 0 & 4     \\
\textbf{carbonGas}                           & 1     & \cellcolor{gray!45} \textbf{0}   & 1                     & \cellcolor{gray!25} 0 & 6     & 0     & \cellcolor{gray!45} \textbf{0}    & {\textbf{\color{gray!75}0}}    & \cellcolor{gray!25} 0 & 3     \\
\textbf{turbine1}                            & 1     & \cellcolor{gray!45} \textbf{0}   & 1                     & \cellcolor{gray!25} 0 & 12    & 1     & \cellcolor{gray!45} \textbf{0}    & 1                     & \cellcolor{gray!25} 0 & 5     \\
\textbf{turbine2}                            & 1     & \cellcolor{gray!45} \textbf{0}   & 1                     & \cellcolor{gray!25} 0 & 8     & 2     & \cellcolor{gray!45} \textbf{0}    & 2                     & \cellcolor{gray!25} 0 & 2     \\
\textbf{turbine3}                            & 1     & \cellcolor{gray!45} \textbf{0}   & 1                     & \cellcolor{gray!25} 0 & 12    & 1     & \cellcolor{gray!45} \textbf{0}    & 1                     & \cellcolor{gray!25} 0 & 5     \\
kepler0                                      & 5     & 4                                & 1                     & \cellcolor{gray!25} 0 & 16    & 6     & 2                                 & 4                     & \cellcolor{gray!25} 0 & 4     \\
cav10                                        & 3     & 1                                & 2                     & \cellcolor{gray!25} 0 & 3     & 1     & \cellcolor{gray!45} \textbf{0}    & 1                     & \cellcolor{gray!25} 0 & 1     \\
\textbf{floudas}                             & 2     & \cellcolor{gray!45} \textbf{0}   & 2                     & \cellcolor{gray!25} 0 & 2     & 1     & \cellcolor{gray!45} \textbf{0}    & 1                     & \cellcolor{gray!25} 0 & 0     \\
\textbf{floudas1}                            & 20    & \cellcolor{gray!45} \textbf{0}   & 20                    & \cellcolor{gray!25} 0 & 36    & 17    & \cellcolor{gray!45} \textbf{0}    & 17                    & \cellcolor{gray!25} 0 & 1     \\
floudas3                                     & 4     & 1                                & 3                     & \cellcolor{gray!25} 0 & 5     & 2     & \cellcolor{gray!45} \textbf{0}    & 2                     & \cellcolor{gray!25} 0 & 1     \\
sum                                          & 4     & 1                                & 3                     & \cellcolor{gray!25} 0 & 15    & 5     & 2                                 & 3                     & \cellcolor{gray!25} 0 & 3     \\
\textbf{doppler(x3)}                         & 9     & \cellcolor{gray!45} \textbf{0}   & 9                     & \cellcolor{gray!25} 0 & 36    & 0     & \cellcolor{gray!45} \textbf{0}    & {\textbf{\color{gray!75}0}}    & \cellcolor{gray!25} 0 & 18    \\
\textbf{rigidBody1}                          & 6     & \cellcolor{gray!45} \textbf{0}   & 6                     & \cellcolor{gray!25} 0 & 14    & 3     & \cellcolor{gray!45} \textbf{0}    & 3                     & \cellcolor{gray!25} 0 & 4     \\
\textbf{himmilbeau}                          & 8     & \cellcolor{gray!45} \textbf{0}   & 8                     & \cellcolor{gray!25} 0 & 14    & 5     & \cellcolor{gray!45} \textbf{0}    & 5                     & \cellcolor{gray!25} 0 & 2     \\
\textbf{matrixDeterminant(x2)}               & 20    & \cellcolor{gray!45} \textbf{0}   & 20                    & \cellcolor{gray!25} 0 & 56    & 10    & \cellcolor{gray!45} \textbf{0}    & 10                    & \cellcolor{gray!25} 0 & 18    \\

\bottomrule
\end{tabular}
}

Programs annotated with “(x2)” or “(x3)” denote groups of two or three similar benchmarks. FPScan col. contains the number of detected pitfalls.
\vspace{-2em}
\end{table}

\begin{figure*}[h]
    \centering
    \includegraphics[width=0.8\textwidth]{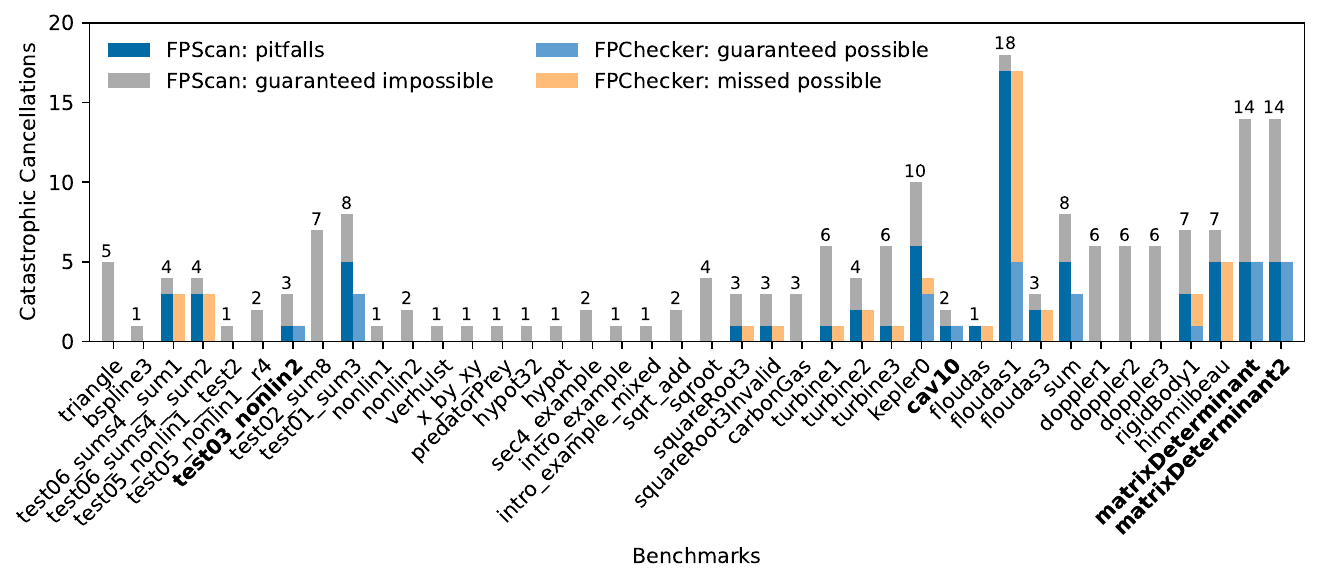}
    \caption{Comparison of catastrophic cancellation detection by FPScan and FPChecker. The FPScan bar is split into cancellations detected by FPScan and cancellations proven impossible by FPScan, while the FPChecker bar is split into cancellations detected by FPChecker and occurring cancellations missed by FPChecker.}\vspace{-1em}
    \label{fig::fps_fpc}
\end{figure*}
\subsection{Experimental Setup}
To answer RQ1, we compare FPScan against a ground truth produced by a sound and complete \textit{bitblasting-based} tool that we developed for this purpose. \textit{Bitblasting} is a technique used by SMT solvers to reason about \fp arithmetic by encoding \fp values as bit vectors, possibly after dedicated optimizations~\cite{brillout_bit_blast}. Although sound and complete for decidable theories, this approach remains too computationally expensive for large-scale practical use. Our implementation relies on Z3~\cite{Z3} and uses the following detection criteria: (i) catastrophic cancellation is reported when all bits of the result are lost, and (ii) absorption is reported when the least significant bit of one operand is smaller than the most significant bit of the other. This definition of catastrophic cancellation is stricter than the one implemented by FPScan. Consequently, some genuine cancellations may be incorrectly tagged as spurious by bitblasting. Although this does not undermine the overall quality of the evaluation, it may lead to an underestimation of FPScan's completeness.

We evaluate FPScan on a subset of the FPBench benchmark suite~\cite{damouche2016toward}, which contains 130 numerical programs drawn from numerical analysis papers and textbooks; computations are performed in \textit{binary32} format.
We retain 58 benchmarks whose preconditions are interval based and compatible with our tool.
Among them, 41 are also compatible with \textit{bitblasting}. 

We also compare FPScan with FPChecker~\cite{laguna_fpchecker_2019, laguna_fpchecker_2022}, a state-of-the-art dynamic tool for detecting \fp catastrophic cancellations. FPChecker reports a cancellation when the number of significant bits in an addition or subtraction falls below 10. Although FPChecker also detects other numerical issues, such as overflow and underflow, it does not detect absorption. We therefore compare FPChecker and FPScan only with respect to cancellation detection. Since FPChecker requires concrete inputs, we sample each input interval using 1000 points. All experiments were conducted on a machine running Ubuntu 22.04.5, equipped with an Intel i7 1370P processor and 16 GB of memory.

\begin{figure*}[tb]
\vspace{-1em}
\makebox[\textwidth][c]{
    \vspace{-1em}
    \includegraphics[width=0.85\textwidth]{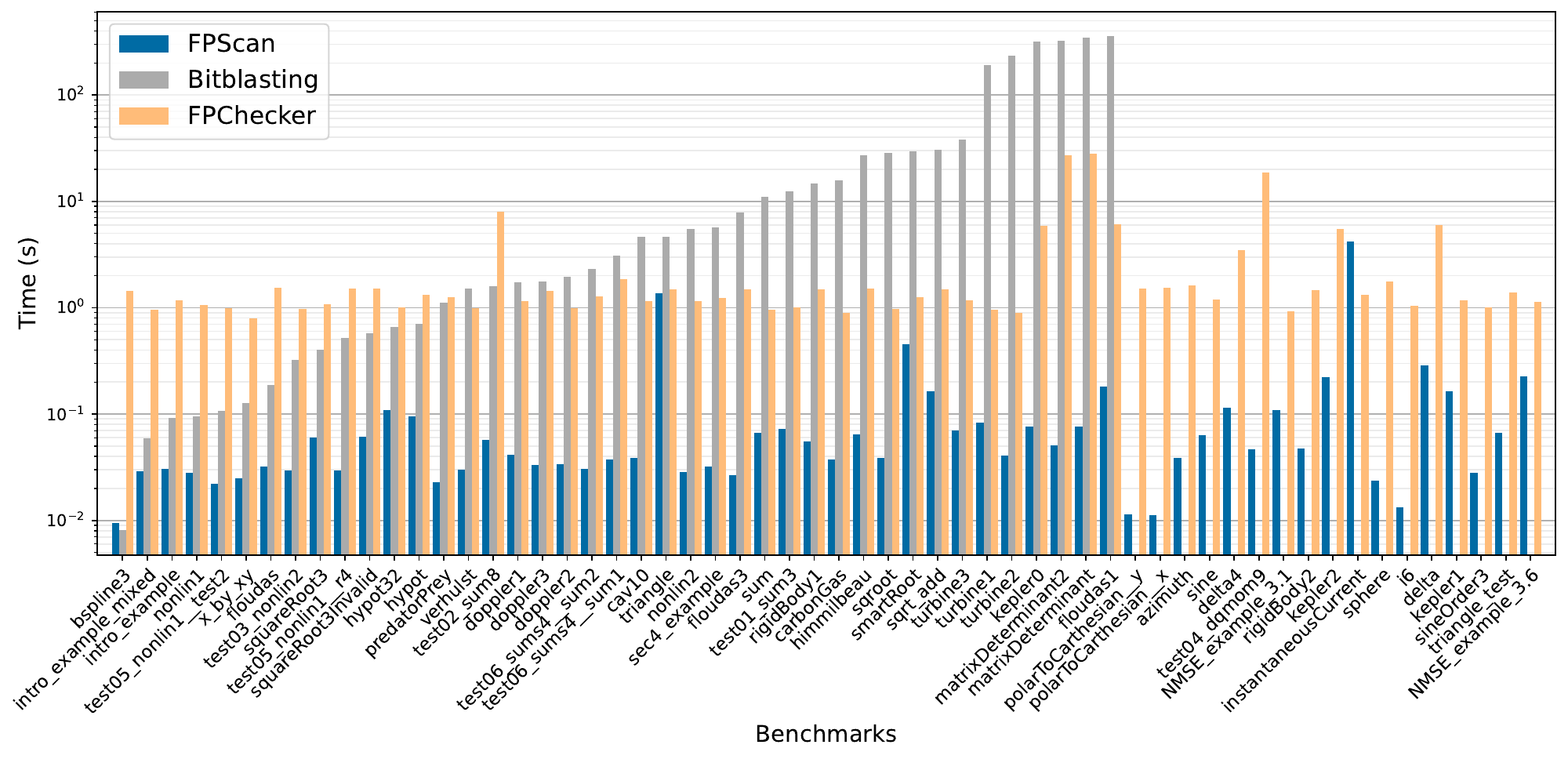}
    }
\vspace{-2em}    \caption{ Computation time (s) for FPScan, FPChecker~(1000 point sampled)~
    \cite{laguna_fpchecker_2019,laguna_fpchecker_2022}, and  \textit{bitblasting}. }
    \vspace{-1em}
    \label{fig::tps_results}
\end{figure*}

\subsection{\textbf{RQ1}: Soundness Validation and Completeness Evaluation.}

Table~\ref{tab::results} reports the comparison between FPScan and the
\textit{bitblasting-based} ground truth. We classify all detected pitfalls into
four categories: true positives (TP), corresponding to pitfalls reported as
possible by FPScan that indeed occur; true negatives (TN), corresponding to
pitfalls proven impossible by FPScan that do not occur; false positives (FP),
corresponding to pitfalls reported as possible by FPScan that are actually
impossible; and false negatives (FN), corresponding to pitfalls proven
impossible by FPScan that actually occur.

FPScan is sound by construction; consequently, the FN columns (highlighted in
light gray) in Table~\ref{tab::results} contain only zeros. This means that
every pitfall proven impossible by FPScan is indeed impossible. 

FPScan is also effective at proving the absence of pitfalls: out of 665 checked
cases, 482 (72\%) are proven impossible.

Among the 41 benchmark analyses, 36 contain no false positives.
Overall, only 14 (8\%) of the 183 reported pitfalls are false positives,
indicating that the number of spurious reports is relatively limited. This
incompleteness is a consequence of FPScan's abstraction of machine states: it
includes states that are not actually reachable. Consequently, FPScan may fail
to prove that some non-occurring pitfalls cannot occur. These false positives
are reported in the FP columns of Table~\ref{tab::results}.

Figure~\ref{fig::fps_fpc} shows a comparison of the detection capabilities of
FPScan and FPChecker. FPChecker fails to detect any catastrophic cancellations
on some programs, due to its sampling strategy, which is effective for certain
programs but less suitable for others. Addressing this limitation by designing a
sampling strategy that is robust across programs is left for future work.
Additionally, Figure~\ref{fig::fps_fpc} graphically illustrates the accuracy of
FPScan. For most programs, the blue portion of the left bar, representing the
catastrophic cancellations detected by FPScan, closely matches the right bar,
which represents all occurring catastrophic cancellations.

Figure~\ref{fig::fps_fpc} also highlights a potential synergy between FPScan and
FPChecker. For all programs shown in bold, the catastrophic cancellations detected
by FPScan and FPChecker exactly match. This is particularly interesting because,
in these cases, catastrophic cancellations are divided into two categories: they
are either proven impossible by FPScan or guaranteed possible by FPChecker
through a counterexample. The two tools provide complementary information.

\subsection{\textbf{RQ2} Execution Time Evaluation}

We measure the execution time of FPScan, FPChecker and \textit{bitblasting} on the 58 benchmarks in the FPBench subset. The results are shown in Figure~\ref{fig::tps_results}. Programs are ordered by increasing \textit{bitblasting} execution time. Some programs have no gray bar because \textit{bitblasting} either timed out or does not support some required operations. Due to the large variation in computation time, a logarithmic scale is used for the vertical axis.

The execution time of FPChecker is proportional to the number of samples used in this experiment, namely 1000. For reference, its execution time is comparable to FPScan with 100 samples.

Three programs exhibit significantly higher execution times, namely \textit{triangle}, \textit{smartRoot}, and \textit{instantaneousCurrent}. This is due to the interaction between the iterative over-approximation of square root~\cite{daumas_verified_2007} and the use of rational numbers in the interval analysis.
Despite issuing one SMT solver query per pitfall check, our method remains efficient. It is comparable to FPChecker and significantly outperforms \textit{bitblasting} in terms of execution time. This efficiency stems from the simple structure of the constraint system, which involves integer variables, linear operations, \texttt{min} and \texttt{max} functions, and a limited number of conjunctions. Although scalability remains a concern, as each statement increases the constraint size and may trigger up to three SMT solver queries, the approach still demonstrates strong practical performance on the evaluated benchmarks. Future work will investigate strategies to improve scalability and reduce the number of solver queries per statement.

\section{Related Work}
\label{section::related}
A large body of work focuses on bounding rounding errors~\cite{Gappa,precisa,fluctuat,astree2002,EVA2017,jezequel_cadna_2008}
or optimizing \fp expressions~\cite{panchekha_automatically_2015,sardana2012,POP2019,precimonious}. 
All these works focus on the characterization of the error, usually computing its upper bound. This is not the focus of this paper.
Our primary objective is to formally prove the absence of \fp pitfalls that can cause severe precision loss. This can help programmers identify weaknesses of \fp expressions.

\subsection{$\ufp$, $\ulp$, and $\nsb$ constraint-based analysis}

Martel~\cite{martel2017floating} gives definitions of $\ufp$ and $\ulp$ that translate order of magnitude of \fp numbers and rounding errors into integer values. 
He uses them to compute errors of constants to propagate them in abstract domains. We rely on a similar approach by formalizing $\ufp$ and $\ulp$ as constraints and extending their definitions to handle denormalized numbers.

Ben~Khalifa~et~al.~\cite{POP2019} provide a sound constraint-based precision-tuning tool. They leverage the definition of $\ufp$ and $\ulp$ by Martel~\cite{martel2017floating} to establish constraints linking the precisions of all variables of a program. Solving these constraints yields a set of intermediate variable precisions that make the output meet a target precision. 
We rely on a similar approach: producing constraints to model program behavior, but our constraint set is completely different.
We also develop a new constraint set including $\nsb$ and refining error propagation. In addition, rather than targeting the satisfiability of these constraints, we target their unsatisfiability to prove the absence of pitfalls.

\subsection{Detection of Floating-Point Pitfalls}
We can divide the state-of-the-art on the detection of floating-point pitfalls in two categories: dynamic and static methods. The former execute the programs, yielding unsound yet complete results. The latter reason about the programs yielding sound yet incomplete results.

\paragraph{Dynamic Methods}
Most existing approaches are dynamic. The program is executed on a set of
inputs. These methods are often fast and scalable to large and complex programs.
However, they are in practice computationally expensive, both in terms of
execution time and memory consumption, especially when a large number of inputs
must be explored. Moreover, their results are only valid for the explored
inputs, with no guarantees over the rest of the input space. In this context,
Lam et al.~\cite{LamHS13}, Benz et al.~\cite{benz},
Jézéquel~et~al.\cite{jezequel_cadna_2008}, and Laguna et
al.~\cite{laguna_fpchecker_2019,laguna_fpchecker_2022} propose similar
approaches. Their tools  analyze programs to estimate the number of bits lost
after each addition and subtraction. A  catastrophic cancellation is reported
when this number drops below a given threshold.

Ben Khalifa et al.~\cite{khalifa_toward_2022} rely on shadow execution at a
higher precision to approximate \fp errors. They distinguish two types of
cancellation: catastrophic cancellation, which occurs when the magnitude of the
accumulated error becomes comparable to or larger than the magnitude of the
computed result during an addition or subtraction, and benign cancellation,
where only part of the significant bits are lost. We formalized their
characterization of catastrophic cancellation as a constraint in our approach.

Overall, these dynamic techniques target cancellations rather than absorptions.
They achieve good performance, but are not sound unless the entire input space
is explored, which is generally infeasible in practice. Our method uses a
constraint-based abstraction of reachable machine states to soundly reason over
all executions and prove the absence of both catastrophic cancellations and
absorptions.

\paragraph{Static Methods}

Static methods infer the behavior of a program over all possible inputs by
analyzing its code rather than executing it. The goal is to efficiently compute
properties that hold for all program executions. To the best of our knowledge,
only Lopes et al.~\cite{lopes2018} address the sound detection of both
absorption and  catastrophic cancellation. Their approach first performs a range
analysis to bound all intermediate variables, and then applies decision rules to
determine whether a numerical pitfall may occur. In their definition, an
absorption  occurs when the difference between the result of an operation and
one of the operands is smaller than the rounding error. However, the paper does
not discuss denormalized numbers, so it is difficult to assess whether the
method extends to them soundly. Moreover, the treatment of loops is not
described in sufficient detail to determine how pitfall conditions are
propagated across iterations. Our approach addresses these issues by supporting
denormalized values, soundly analysing loops up to a given unrolling bound, and
taking rounding errors into account in the catastrophic cancellation detection
criterion. However, we could not directly compare our approach with
their tool because it is not freely available.

\section{Conclusion}
\label{section::conclusion}
We presented FPScan, a novel tool and approach to formally detect both
catastrophic cancellation and absorption in \fp programs. Our approach combines
a range analysis using abstract interpretation with an axiomatization as a set
of logical constraints of the propagation of numerical errors in the program.
Proving the absence of \fp pitfalls amount to prove that a formula is
unsatisfiable with an SMT solver. Experimental results show that FPScan is both
precise and efficient, proving 72\% of cases impossible with only 8\% false
positives, while outperforming bitblasting and comparing favorably with
FPChecker.

%
Future work will extend FPScan to mixed-precision programs.
We also plan to rely on identified pitfalls to propose guaranteed repairs, for
example, using precision tuning~\cite{POP2019}, expression
transformations~\cite{panchekha_automatically_2015}, and Taylor
expansions~\cite{panchekha_automatically_2015}. 

\bibliographystyle{IEEEtran}
\bibliography{bibli.bib}
\newpage
\appendix
In this section, we provide the proofs of the theorems presented in Section~\ref{section::detection}.
We first recall the soundness proof for Addition/Subtraction.

Note that the proofs are valid in the two edge cases of zero and denormalized numbers. This follows from the relations between rounding error and $\ulp$, and between order of magnitude and $\ufp$ established by Eqs.\eqref{def::ufp::bound} and \eqref{def::ulp::error}, respectively. Both equations are valid for zero and denormalized numbers.

\subsection{Soundness of abstract addition/subtraction (cf. Def.~\ref{cstr::add::theorem})}
\label{section::annex::proof::add_sub}

Let $\varepsilon_x$, $\varepsilon_y$, and $\varepsilon_z$ denote the \fp errors: 
\begin{equation}\small
    \begin{array}{l}
        \varepsilon_x = x - \tilde{x} \Leftrightarrow \tilde{x} = x - \varepsilon_x \Leftrightarrow x = \tilde{x} + \varepsilon_x  \enspace ; \\
        \varepsilon_y = y - \tilde{y} \Leftrightarrow \tilde{y} = y - \varepsilon_y \Leftrightarrow y = \tilde{y} + \varepsilon_y  \enspace ; \\
        \varepsilon_z = z - \tilde{z} \Leftrightarrow \tilde{z} = z - \varepsilon_z \Leftrightarrow z = \tilde{z} + \varepsilon_z \enspace .  \\
    \end{array}
\end{equation}
By applying Eq.~\eqref{def::nsb::eq}, a bound on \( \nsb_{\sigma}(z_v) \)
can be derived from a bound on \( \ufp(\varepsilon_z) \). This bound,
in turn, can be obtained by bounding \( |\varepsilon_z| \) using
Eq.~\eqref{def::ufp::bound}. An explicit expression for \( \varepsilon_z \)
is obtained through a series of manipulations, as demonstrated in
Eq.~\eqref{cstr::add::proof::epsilon_z}. Here, we denote by \( e_+ \) the
rounding error introduced by the addition.

\begin{equation}\small
\begin{aligned}
\tilde{z} &= \round(\tilde{x} + \tilde{y}) = \tilde{x} + \tilde{y} + e_+ \\
          &= (x - \varepsilon_x) + (y - \varepsilon_y) + e_+
           = \underbrace{x + y}_{z} - (\underbrace{\varepsilon_x + \varepsilon_y - e_+}_{\varepsilon_z})
\end{aligned}
\label{cstr::add::proof::epsilon_z}
\end{equation}

 The error $|\varepsilon_z|$ can be bounded using the triangle inequality as shown

\begin{equation}
    |\varepsilon_z| \leq |\varepsilon_x| + |\varepsilon_y| + |e_+| \enspace .
    \label{cstr::add::proof::abs_epsilon_z}
\end{equation}

The errors $|\varepsilon_x|$ and $|\varepsilon_y|$ can be bounded using
Eq.~\eqref{def::ufp::bound} and $e_+$ is a simple rounding error; thus,
$|e_+| \leq 2^{\ulp(\tilde{z}) - 1}$ using Eq.~\eqref{def::ulp::error}.
Hence, Eq.~\eqref{cstr::add::proof::abs_epsilon_z} can be rewritten as
Eq.~\ref{cstr::add::proof::3_2_epsilon_z}.

\begin{equation}
    |\varepsilon_z| < 2^{\ufp(\varepsilon_x) + 1} + 2^{\ufp(\varepsilon_y) + 1} + 2^{\ulp(\tilde{z}) - 1} \enspace .
    \label{cstr::add::proof::3_2_epsilon_z}
\end{equation}

Eq.~\eqref{cstr::add::proof::3_2_epsilon_z} is the sum of three powers of
two. It can thus be bounded as shown in Eq.~\eqref{cstr::add::proof::2_epsilon_z}.

\begin{equation}
    |\varepsilon_z| < 2^{\max(\ufp(\varepsilon_x) + 1, \ufp(\varepsilon_y) + 1, \ulp(\tilde{z})-1) + 2} \enspace .
    \label{cstr::add::proof::2_epsilon_z}
\end{equation}

We can deduce $\ufp(\varepsilon_z)$ from Eq.~\eqref{cstr::add::proof::2_epsilon_z} using Eq.~\eqref{def::ufp::bound}
and thus obtain

\begin{equation}
    \ufp(\varepsilon_z) \leq \max(\ufp(\varepsilon_x) + 1, \ufp(\varepsilon_y) + 1, \ulp(\tilde{z})-1) + 1 \enspace .
    \label{cstr::add::proof::ufp_epsilon_z}
\end{equation}

Eq.~\eqref{cstr::add::proof::ufp_epsilon_z} together with Eq.~\eqref{def::nsb::eq} allows us to obtain the desired result.

\subsection{Soundness of abstract multiplication (cf. Def.~\ref{cstr::mult::theorem})}
\label{section::annex::proof::mult}
Let $\sigma$ denotes a machine state and $x_v,~y_v,~z_v$ denote three variables.
$x$, $y$ and $z$ denotes their respective real values, $x =
\sigma_\mathbb{R}(x_v)$, $y = \sigma_\mathbb{R}(y_v)$ and $z =
\sigma_\mathbb{R}(z_v)$. $\tilde{x}$, $\tilde{y}$ and $\tilde{z}$ denote their
respective \fp values, $\tilde{x} = \sigma_{\mathbb{F}_p}(x_v)$,
$\tilde{y} = \sigma_{\mathbb{F}_p}(y_v)$ and $\tilde{z} =
\sigma_{\mathbb{F}_p}(z_v)$.

\begin{equation}
    \begin{array}{l}
        \varepsilon_x = x - \tilde{x} \Leftrightarrow \tilde{x} = x - \varepsilon_x \Leftrightarrow x = \tilde{x} + \varepsilon_x  \enspace ;  \\
        \varepsilon_y = y - \tilde{y} \Leftrightarrow \tilde{y} = y - \varepsilon_y \Leftrightarrow y = \tilde{y} + \varepsilon_y    \enspace ;\\
        \varepsilon_z = z - \tilde{z} \Leftrightarrow \tilde{z} = z - \varepsilon_z \Leftrightarrow z = \tilde{z} + \varepsilon_z    \enspace .\\
    \end{array}
\end{equation}

Using the same reasoning as in the proof of Section~\ref{section::annex::proof::add_sub}, the constraint will bound
$\nsb_\sigma(z_v)$. Eq.~(\ref{cstr::mult::proof::epslion_z}) shows how to
bound $\varepsilon_z$. We denote by $e_\times$ the rounding error introduced by
the multiplication.

\begin{align}
    & z = x \times y \nonumber\\
    & z = (\tilde{x} + \varepsilon_x) \times (\tilde{y} + \varepsilon_y) \nonumber\\
    & z = (\tilde{x} \times \tilde{y}) + \tilde{x}\varepsilon_y + \tilde{y}\varepsilon_x + \varepsilon_x\varepsilon_y \nonumber\\
    & z = \underbrace{\round(\tilde{x} \times \tilde{y})}_{\tilde{z}} + (\underbrace{\tilde{x}\varepsilon_y + \tilde{y}\varepsilon_x + \varepsilon_x\varepsilon_y - e_\times}_{\varepsilon_z}) \label{cstr::mult::proof::epslion_z}
\end{align}

The error $|\varepsilon_z|$ can be bounded using triangle inequality:

\begin{equation}
    |\varepsilon_z| \leq
        \underbrace{|\tilde{x}\varepsilon_y|}_{b_1}
        + \underbrace{|\tilde{y}\varepsilon_x|}_{b_2}
        + \underbrace{|\varepsilon_x\varepsilon_y|}_{b_3}
        + |e_\times| \enspace .
    \label{cstr::mult::proof::epsilon_z_a}
\end{equation}

The digits $b_1$, $b_2$, $b_3$ and the error $e_\times$ can be bounded independently. Equations
(\ref{cstr::mult::proof::b1::1}), (\ref{cstr::mult::proof::b1::2}) and
(\ref{cstr::mult::proof::b3::1}) can be obtained using equations
(\ref{def::ufp::bound}) and (\ref{def::nsb::eq}).
\begin{align}
    & |\tilde{x}| < 2^{\ufp(\tilde{x}) + 1} \label{cstr::mult::proof::b1::1} \enspace .\\
    & |\tilde{y}| < 2^{\ufp(\tilde{y}) + 1} \label{cstr::mult::proof::bi::2} \enspace .\\
    & |\varepsilon_y| < 2^{\ufp(\tilde{y}) - \nsb_\sigma(y_v) + 1} \label{cstr::mult::proof::b1::2} \enspace .\\
    & |\varepsilon_x| < 2^{\ufp(\tilde{x}) - \nsb_\sigma(x_v) + 1} \label{cstr::mult::proof::b3::1} \enspace .
\end{align}

Together, equations (\ref{cstr::mult::proof::b1::1}) and
(\ref{cstr::mult::proof::b1::2}) allow to bound $b_1$ as shown in Equations
(\ref{cstr::mult::proof::b1}). Eq. (\ref{cstr::mult::proof::b2}) is obtained
using a similar reasoning.

\begin{align}
    & b_1 = |\tilde{x}\varepsilon_y| < 2^{\ufp(\tilde{x}) + \ufp(\tilde{y}) - \nsb_\sigma(y_v) + 2} \label{cstr::mult::proof::b1} \enspace .\\
    & b_2 = |\tilde{y}\varepsilon_x| < 2^{\ufp(\tilde{x}) + \ufp(\tilde{y}) - \nsb_\sigma(x_v) + 2} \label{cstr::mult::proof::b2} \enspace .
\end{align}

Eq. (\ref{cstr::mult::proof::b3}) is obtained using Equations
(\ref{cstr::mult::proof::b1::2}) and (\ref{cstr::mult::proof::b3::1}).

\begin{equation}
    b_3 = |\varepsilon_x\varepsilon_y| < 2^{\ufp(\tilde{x}) + \ufp(\tilde{y}) - \nsb_\sigma(x_v) - \nsb_\sigma(y_v) + 2} \label{cstr::mult::proof::b3} \enspace .
\end{equation}

Finally, we obtain Eq. (\ref{cstr::mult::proof::e_x}) using Eq.
(\ref{def::ulp::error}) since $e_\times$ is a rounding error:
\begin{equation}
    |e_\times| \leq 2^{\ulp(\tilde{z}) - 1} \enspace .
    \label{cstr::mult::proof::e_x}
\end{equation}

Eq. (\ref{cstr::mult::proof::epsilon_z_a}) can be rewritten using
equations (\ref{cstr::mult::proof::b1}) to  (\ref{cstr::mult::proof::e_x}) to  finally obtain
Eq. (\ref{cstr::mult::proof::4_2_epsilon_z}).

\begin{equation}
    \begin{array}{ll}
        |\varepsilon_z| <   & ~ 2^{\ufp(\tilde{x}) + \ufp(\tilde{y}) - \nsb_\sigma(y_v) + 2} \\
        & + 2^{\ufp(\tilde{x}) + \ufp(\tilde{y}) - \nsb_\sigma(x_v) + 2} \\
        & + 2^{\ufp(\tilde{x}) + \ufp(\tilde{y}) - \nsb_\sigma(x_v) - \nsb_\sigma(y_v) + 2} \\
        & + 2^{\ulp(\tilde{z}) - 1}
    \end{array} \enspace .
    \label{cstr::mult::proof::4_2_epsilon_z}
\end{equation}

Since Eq.~(\ref{cstr::mult::proof::4_2_epsilon_z}) is the sum of four terms, each a power of two, it can be bounded by a single power of two. Let \( M \) denote the maximum exponent, as defined in Eq.~(\ref{cstr::mult::proof::M}).

\begin{equation}\small
    \resizebox{0.9\linewidth}{!}{$
    \begin{array}{rl}
        M &=  \max\left(\begin{array}{l}
                                    \ufp(\tilde{x}) + \ufp(\tilde{y}) - \nsb_\sigma(y_v) + 2   \\
                                    \ufp(\tilde{x}) + \ufp(\tilde{y}) - \nsb_\sigma(x_v) + 2   \\
                                    \ufp(\tilde{x}) + \ufp(\tilde{y}) - \nsb_\sigma(x_v) - \nsb_\sigma(y_v) + 2 \\
                                    \ulp(\tilde{z}) - 1
                                \end{array}\right)\\
                            &= \max\left(\begin{array}{l}
                                    \ufp(\tilde{x}) + \ufp(\tilde{y}) - \min(\nsb_\sigma(x_v), \nsb_\sigma(y_v)) + 2 \\
                                    \ufp(\tilde{x}) + \ufp(\tilde{y}) - \nsb_\sigma(x_v) - \nsb_\sigma(y_v) + 2 \\
                                    \ulp(\tilde{z}) - 1
                                \end{array}\right)
    \end{array}
    $} \enspace .
    \label{cstr::mult::proof::M}
\end{equation}
It is thus possible to establish a bound on \( |\varepsilon_z| \). Given the maximum exponent \( M \), this upper bound is attained when all four exponents are equal. Under such circumstances, two carry bits may arise, which yields \( |\varepsilon_z| < 2^{M+2} \). Thus, using Eq.~(\ref{def::ufp::bound}) we obtain \( \ufp(\varepsilon_z) \leq M + 1 \) and then Eq.~(\ref{cstr::mult::proof::end}) replacing \( M \) by its value.

\begin{equation}\small
\resizebox{0.9\linewidth}{!}{$
\ufp(\varepsilon_z) \leq \max\left(\begin{array}{l}
        \ufp(\tilde{x}) + \ufp(\tilde{y}) - \min(\nsb_\sigma(x_v), \nsb_\sigma(y_v)) + 3 \\
        \ufp(\tilde{x}) + \ufp(\tilde{y}) - \nsb_\sigma(x_v) - \nsb_\sigma(y_v) + 3 \\
        \ulp(\tilde{z})
    \end{array}\right)
    \label{cstr::mult::proof::end}
$} \enspace .
\end{equation}

\noindent
Equations (\ref{cstr::mult::proof::end}) and (\ref{def::nsb::eq}) allow to conclude.

\hfill $\blacksquare$


\subsection{Soundness of abstract division (cf. Def.~\ref{cstr::div::theorem})}
\label{section::annex::proof::div}

Let $\sigma$ denotes a machine state and $x_v,~y_v,~z_v$ denote three variables.
$x$, $y$ and $z$ denotes their respective real values, $x =
\sigma_\mathbb{R}(x_v)$, $y = \sigma_\mathbb{R}(y_v)$ and $z =
\sigma_\mathbb{R}(z_v)$. $\tilde{x}$, $\tilde{y}$ and $\tilde{z}$ denote their
respective \fp values, $\tilde{x} = \sigma_{\mathbb{F}_p}(x_v)$,
$\tilde{y} = \sigma_{\mathbb{F}_p}(y_v)$ and $\tilde{z} =
\sigma_{\mathbb{F}_p}(z_v)$

\begin{equation}
    \begin{array}{l}
        \varepsilon_x = x - \tilde{x} \Leftrightarrow \tilde{x} = x - \varepsilon_x \Leftrightarrow x = \tilde{x} + \varepsilon_x \enspace ;   \\
        \varepsilon_y = y - \tilde{y} \Leftrightarrow \tilde{y} = y - \varepsilon_y \Leftrightarrow y = \tilde{y} + \varepsilon_y \enspace ;  \\
        \varepsilon_z = z - \tilde{z} \Leftrightarrow \tilde{z} = z - \varepsilon_z \Leftrightarrow z = \tilde{z} + \varepsilon_z  \enspace .  \\
    \end{array}
\end{equation}

Using the same reasoning as in the soundness proof of Section~\ref{section::annex::proof::add_sub}, the constraint will bound
$\nsb_\sigma(z_v)$. Eq.~(\ref{cstr::div::proof::epslion_z}) shows how to
bound $\varepsilon_z$. We denote by $e_\div$ the rounding error introduced by
the division.

\begin{align}\small
    \label{cstr::div::proof::epsilon_z}
    \varepsilon_z
        &= \frac{x}{y} - \round(\frac{\tilde{x}}{\tilde{y}})                \nonumber   \\
        &= \frac{x}{y} - \frac{\tilde{x}}{\tilde{y}} - e_{\div}             \nonumber   \\                                                      \\
        &=  \frac{\tilde{x} + \varepsilon_x}{\tilde{y} + \varepsilon_y}
            - \frac{\tilde{x}}{\tilde{y}}
            - e_{\div}                                                      \nonumber   \\
        &=  \frac{
                \tilde{y}(\tilde{x} + \varepsilon_x)
                - \tilde{x}(\tilde{y} + \varepsilon_y)
            }{
                \tilde{y}(\tilde{y} + \varepsilon_y)
            }
            - e_{\div}                                                                  \\
        &=  \frac{
                \tilde{y}\varepsilon_x - \tilde{x}\varepsilon_y
            }{
                \tilde{y}(\tilde{y} + \varepsilon_y)
            }
            - e_{\div} \label{cstr::div::proof::epslion_z} \enspace .
\end{align}

We obtain Eq. (\ref{cstr::div::proof::e_/}) using Eq.
(\ref{def::ulp::error}) since $e_\div$ is a rounding error:
\begin{equation}
    |e_\times| \leq 2^{\ulp(\tilde{z}) - 1} \enspace .
    \label{cstr::div::proof::e_/} \enspace .
\end{equation}

The error $|\varepsilon_z|$ can be bound using the triangle inequality as shown:
\begin{align}
    |\varepsilon_z|
        &=  |\frac{
                \tilde{y}\varepsilon_x - \tilde{x}\varepsilon_y
            }{
                \tilde{y}(\tilde{y} + \varepsilon_y)
            }
            - e_{\div}|                                             \nonumber\\
        &\leq
            |\frac{
                \tilde{y}\varepsilon_x - \tilde{x}\varepsilon_y
            }{
                \tilde{y}(\tilde{y} + \varepsilon_y)
            }|
            + |e_{\div}|                                            \nonumber\\
        &\leq
            |\frac{
                \tilde{y}\varepsilon_x - \tilde{x}\varepsilon_y
            }{
                \tilde{y}(\tilde{y} + \varepsilon_y)
            }|
            + 2^{\ulp(z) - 1}      \enspace .  \label{cstr::div::proof::bound_epsilon}
\end{align}

To bound $|\varepsilon_z|$, one therefore needs to find an upper bound of
$|\tilde{y}\varepsilon_x - \tilde{x}\varepsilon_y|$ and a lower bound of
$|\tilde{y}(\tilde{y} + \varepsilon_y)|$.

\begin{align}
    |\tilde{x}|       &< 2^{\ufp(\tilde{x}) + 1}              \label{cstr::div::proof::bound_x}  \enspace .  \\
    |\tilde{y}|       &< 2^{\ufp(\tilde{y}) + 1}              \label{cstr::div::proof::bound_y}  \enspace .  \\
    |\varepsilon_x|   &< 2^{\ufp(\tilde{x}) - \nsb_\sigma(x_v) + 1}    \label{cstr::div::proof::bound_ex} \enspace .  \\
    |\varepsilon_y|   &< 2^{\ufp(\tilde{y}) - \nsb_\sigma(y_v) + 1}    \label{cstr::div::proof::bound_ey} \enspace .
\end{align}

\sloppy
Together, Equations (\ref{cstr::div::proof::bound_x}),
(\ref{cstr::div::proof::bound_y}), (\ref{cstr::div::proof::bound_ex}), and
(\ref{cstr::div::proof::bound_ey}) allow us to bound \mbox{$|\tilde{y}\varepsilon_x -
\tilde{x}\varepsilon_y|$} as shown in Eq.
(\ref{cstr::div::proof::bound_1_1}). As Eq.
(\ref{cstr::div::proof::bound_1_1}) is made of two power of two, the upper bound
can be refined as shown in Eq. (\ref{cstr::div::proof::bound_1}).
\fussy

\begin{align}
    |\tilde{y}\varepsilon_x - \tilde{x}\varepsilon_y|
        &<
            2^{\ufp(\tilde{x}) + \ufp(\tilde{y}) - \nsb_\sigma(x_v) + 2} \nonumber\\
        &\quad
            + 2^{\ufp(\tilde{x}) + \ufp(\tilde{y}) - \nsb_\sigma(y_v) + 2}               \label{cstr::div::proof::bound_1_1}\\
        &<
            2^{\ufp(\tilde{x}) + \ufp(\tilde{y}) - \min(\nsb(x_v), \nsb_\sigma(y_v)) + 3}  \label{cstr::div::proof::bound_1}\enspace .
\end{align}

Now, let us focus on a lower bound for $|\tilde{y}(\tilde{y} + \varepsilon_y)|$.
The main issue is that $\tilde{y} + \varepsilon_y$ can be close to $0$ for large
errors. We limit the analysis to scenarios with small errors, \textit{i.e.} when
$\ufp(\tilde{y}) > \ufp(\varepsilon_y) + 1$ which can be simplified to $\nsb_\sigma (y_v) > 1$.
Otherwise, no constraint is applied.

To use Equations (\ref{cstr::div::proof::bound_y}) and
(\ref{cstr::div::proof::bound_ey}) allows to get Eq.
(\ref{cstr::div::proof::bound_denom})

\begin{align}
    |\tilde{y} + \varepsilon_y| &   \geq ||\tilde{y}| - |\varepsilon_y||    \nonumber\\
                                &   \geq |\tilde{y}| - |\varepsilon_y|      \nonumber\\
                                &   > 2^{\ufp(\tilde{y})} - 2^{\ufp(\varepsilon_y) + 1}   \nonumber\\
                                &   > 2^{\ufp(\tilde{y}) - 1}   \label{cstr::div::proof::bound_denom}
\end{align}

Eq. (\ref{cstr::div::proof::bound_denom}) allows to conclude
$|\tilde{y}(\tilde{y} + \varepsilon_y)| > 2^{2\ufp(\tilde{y}) - 1}$ and then to get Eq. (\ref{cstr::div::proof::bound_left}):

\begin{align}
    |\frac{
        \tilde{y}\varepsilon_x - \tilde{x}\varepsilon_y
    }{
        \tilde{y}(\tilde{y} + \varepsilon_y)
    }|
        & < \frac{2^{\ufp(\tilde{x}) + \ufp(\tilde{y}) - \min(\nsb_\sigma(x_v), \nsb_\sigma(y_v)) + 3}}{2^{2\ufp(\tilde{y}) - 1}} \nonumber\\
        & < 2^{\ufp(\tilde{x}) - \ufp(\tilde{y}) - \min(\nsb_\sigma(x_v), \nsb_\sigma(y_v)) + 4} \label{cstr::div::proof::bound_left} \enspace .
\end{align}

\sloppy
Equations (\ref{cstr::div::proof::bound_epsilon}) and
(\ref{cstr::div::proof::bound_left}) can be put together to obtain
Eq.~(\ref{cstr::div::proof::bound_epsilon_2_sum}):
\fussy

\begin{align}
    |\varepsilon_z| &< 2^{\ufp(\tilde{x}) - \ufp(\tilde{y}) - \min(\nsb_\sigma(x_v), \nsb_\sigma(y_v)) + 4} + 2^{\ulp(\tilde{z}) - 1} \enspace .
    \label{cstr::div::proof::bound_epsilon_2_sum}
\end{align}

Eq. (\ref{cstr::div::proof::bound_epsilon_2_sum}) is the sum of two power
of two and can thus be bounded as shown in Eq.~\ref{cstr::div::proof::bound_epsilon_2_ufp}:

\begin{equation}
    \resizebox{0.85\linewidth}{!}{$
    |\varepsilon_z| < 2^{
        \max\left(
        \begin{array}{l}
            \ufp(\tilde{x}) - \ufp(\tilde{y}) - \min(\nsb_\sigma(x_v), \nsb_\sigma(y_v)) + 4 \\
            \ulp(\tilde{z}) - 1
        \end{array}
        \right) + 1
        \label{cstr::div::proof::bound_epsilon_2_ufp}
    }
    $}\enspace .
\end{equation}

Which finally allow to conclude with Eq. (\ref{cstr::div::proof::end})

\begin{equation}
    \resizebox{0.85\linewidth}{!}{$
    \begin{array}{l}
        \mathbf{If~} \nsb_\sigma(y_v) > 1 \mathbf{~then~} \\
        \ufp(\varepsilon_z) \leq \max\left(
        \begin{array}{l}
            \ufp(\tilde{x}) - \ufp(\tilde{y}) - \min(\nsb_\sigma(x_v), \nsb_\sigma(y_v)) + 4 \\
            \ulp(\tilde{z}) - 1
        \end{array}
        \right)
    \end{array}
    $}
    \label{cstr::div::proof::end}
\end{equation}
\hfill $\blacksquare$


\subsection{Soundness of abstract Cos, Sin, and Atan (cf. Def.~\ref{cstr::trigo::theorem})}
\label{section::annex::proof::trigo}

Let $\varepsilon_x$, and $\varepsilon_z$ denote the \fp errors:
\begin{equation}\small
\begin{array}{l}
        \varepsilon_x = x - \tilde{x} \Leftrightarrow \tilde{x} = x - \varepsilon_x \Leftrightarrow x = \tilde{x} + \varepsilon_x  \enspace ;  \\
        \varepsilon_z = z - \tilde{z} \Leftrightarrow \tilde{z} = z - \varepsilon_z \Leftrightarrow z = \tilde{z} + \varepsilon_z  \enspace .  \\
\end{array}
\end{equation}

Using the same reasoning as in the soundness proof of Section~\ref{section::annex::proof::add_sub}, the constraint will bound $\nsb_\sigma(z_v)$.
Let $f$ denote either sin, cos, or atan. Let $\tilde{z} = \round(f(\tilde{x}))$, $z = f(x)$, and $e$ a rounding error.

{\small
\begin{align}
\tilde{z}
&= \round(f(\tilde{x})) = f(\tilde{x}) + e = f(x - \varepsilon_x) + e \nonumber\\
&= f(x) + (\underbrace{f(x - \varepsilon_x) - f(x)}_R) + e
 = \underbrace{f(x)}_z - \underbrace{(-R - e)}_{\varepsilon_z} \enspace .
\label{cstr::trigo::proof::prev_taylor}
\end{align}
}
We denote by  $R$ the remainder that can be bounded using Taylor-Lagrange Inequality as shown in
Eq.~(\ref{cstr::trigo::proof::taylor}). Note that, since $f$ is either $\mathsf{cos}$, $\mathsf{sin}$, or $\mathsf{atan}$, $\max|f'| = 1$.

\begin{equation}\small
|R| \leq \max|f'|\,|\varepsilon_x| \leq |\varepsilon_x|
< 2^{\ufp(x) - \nsb_\sigma(x) + 1} \enspace .
\label{cstr::trigo::proof::taylor}
\end{equation}

We denote by $e$  a simple rounding error thus
$|e| \leq 2^{\ulp(\tilde{z}) - 1}$ using Eq.~(\ref{def::ulp::error}). Equations~(\ref{cstr::trigo::proof::prev_taylor}) and~(\ref{cstr::trigo::proof::taylor}) can be used to bound $|\varepsilon_z|$: 

{\small
\begin{align}
|\varepsilon_z|
&= |R + e| \le |R| + |e| < 2^{\ufp(x) - \nsb_\sigma(x) + 1} + 2^{\ulp(z) - 1} \nonumber\\
&< 2^{\max(\ufp(x) - \nsb_\sigma(x) + 1,\; \ulp(z) - 1) + 1} \enspace .
\label{cstr::trigo::proof::bound_epsilon_2}
\end{align}
}
Equations~(\ref{def::ufp::bound}) and~(\ref{cstr::trigo::proof::bound_epsilon_2}) allow us to bound $\ufp(\varepsilon_z)$ as shown 

\begin{equation}\small
\ufp(\varepsilon_z) \le
\max\left(\begin{array}{l}
\ufp(x) - \nsb_\sigma(x) + 1 \\
\ulp(z) - 1
\end{array}\right) \enspace .
\end{equation}
\hfill $\blacksquare$


\subsection{Soundness of abstract Sqrt (cf. Def.~\ref{cstr::sqrt::theorem})}
\label{section::annex::proof::sqrt}
Let $\sigma$ denotes a machine state and $x_v,~z_v$ denote two variables.
$x$, and $z$ denotes their respective real values, $x =
\sigma_\mathbb{R}(x_v)$,  and $z =
\sigma_\mathbb{R}(z_v)$. $\tilde{x}$ and $\tilde{z}$ denote their
respective \fp values, $\tilde{x} = \sigma_{\mathbb{F}_p}(x_v)$, and $\tilde{z} =
\sigma_{\mathbb{F}_p}(z_v)$.

\begin{equation}
    \begin{array}{l}
        \varepsilon_x = x - \tilde{x} \Leftrightarrow \tilde{x} = x - \varepsilon_x \Leftrightarrow x = \tilde{x} + \varepsilon_x \enspace .   \\
        \varepsilon_z = z - \tilde{z} \Leftrightarrow \tilde{z} = z - \varepsilon_z \Leftrightarrow z = \tilde{z} + \varepsilon_z \enspace .  \\
    \end{array}
\end{equation}

Using the same reasoning as in the soundness proof of Section~\ref{section::annex::proof::add_sub}, the constraint will bound
$\nsb_\sigma(z_v)$. Eq.~(\ref{cstr::cstr::proof::epslion_z}) shows how to
bound $\varepsilon_z$. We denote by $e$ the rounding error introduced by
the operation.

\begin{align}
    \varepsilon_z   &= z - \tilde{z}    \nonumber\\
                    &= \sqrt{x} - \round(\sqrt{x + \varepsilon_x})  \nonumber\\
                    &= \sqrt{x} - \sqrt{x - \varepsilon_x} - e      \nonumber\\
    |\varepsilon_z| &\leq \underbrace{|\sqrt{x} - \sqrt{x + \varepsilon_x}|}_\Delta + |e| \label{cstr::cstr::proof::epslion_z} \enspace .
\end{align}

As stated by Moscato~et~al.~\cite{precisa}, $\Delta \leq \sqrt{|\varepsilon_x|}$. This result can be obtain computing the derivative of $\Delta$ with respect to $x$.

Using Eq.~(\ref{def::ufp::bound}) bounding the $\ufp$ together with Def.~\ref{def::nsb} defining the $\nsb$, we obtain

\begin{equation}
    |\varepsilon_x| < 2 ^{\ufp(\tilde{x}) - \nsb_\sigma(x_v) + 1}
\end{equation}

that allows to bound $\Delta$:

\begin{align}
    \Delta  &\leq \sqrt{|\varepsilon_x|} \nonumber\\
            &< 2^{\frac{\ufp(\tilde{x}) - \nsb_\sigma(x_v) + 1}{2}}\label{proof::sqrt::delta}\,.
\end{align}

Eq.~(\ref{def::ulp::error}) allows us to bound the rounding error $e$ using $\ulp(\tilde{z})$.

\begin{equation}
    |e| \leq 2^{\ulp(\tilde{z}) - 1}\label{proof::sqrt::e}\,.
\end{equation}

Combining Eq.~(\ref{proof::sqrt::delta})~and~(\ref{proof::sqrt::e}) leads to

\begin{equation}
    |\varepsilon_z| < 2^{\frac{\ufp(\tilde{x}) - \nsb_\sigma(x_v) + 1}{2}} + 2^{\ulp(\tilde{z}) - 1}
\end{equation}

which is the sum of two power of two, hence:
\begin{align}
    |\varepsilon_z| &< 2^{
        \max (\frac{\ufp(\tilde{x}) - \nsb_\sigma(x_v) + 1}{2}, \ulp(\tilde{z}) - 1)) + 1
    }
        \label{proof::sqrt::eps}
\end{align}

Eq.~(\ref{proof::sqrt::eps}) bounds $|\varepsilon_z|$ using a power of two. Therefore, we use Eq.~(\ref{def::ufp::bound}) to deduce a bound on $\ufp(\varepsilon_z)$:

\begin{equation}
    \ufp(\varepsilon_z) \leq \max (\frac{\ufp(\tilde{x}) - \nsb_\sigma(x_v) + 1}{2}, \ulp(\tilde{z}) - 1)\label{proof::sqrt::ufp_eps}
\end{equation}

Eq.~(\ref{proof::sqrt::ufp_eps}) combined with Def.~(\ref{def::nsb}) allows us to conclude.

\hfill$\blacksquare$

\end{document}